\documentclass[11pt]{article}

\usepackage[letterpaper,margin=1.00in]{geometry}
\usepackage{amsmath, amssymb, amsthm, amsfonts}
\usepackage{cite}
\usepackage{appendix}
\usepackage{graphicx}
\usepackage{multirow}
\usepackage{ragged2e}
\usepackage{color}
\usepackage{algorithm}
\usepackage[noend]{algpseudocode}
\usepackage{epstopdf}
\usepackage{wrapfig}
\usepackage[textsize=tiny]{todonotes}

\usepackage{caption}
\usepackage{subcaption}
\usepackage{array}

\newcolumntype{P}[1]{>{\centering\arraybackslash}p{#1}} % horizontally centered table cells (i.e. {|P|P|P|})
\newcolumntype{M}{>{\centering\arraybackslash}m} % also vertically centered (i.e. {|M|M|M|})
\usepackage[framemethod=tikz]{mdframed}
\usepackage[font=small]{caption}
\usepackage{xspace}
\usepackage{standalone}

\newtheorem{theorem}{Theorem}[section]
\newtheorem{lemma}[theorem]{Lemma}
\newtheorem{claim}[theorem]{Claim}
\newtheorem{remark}[theorem]{Remark}
\newtheorem*{theorem*}{Theorem}

\newtheorem{definition}[theorem]{Definition}
\newtheorem{question}{Question}

\definecolor{darkgreen}{rgb}{0,0.5,0}
\usepackage{hyperref}
\hypersetup{
    unicode=false,          % non-Latin characters in Acrobat?‚Ç¨‚Ñ¢s bookmarks
    colorlinks=true,        % false: boxed links; true: colored links
    linkcolor=red,          % color of internal links (change box color with linkbordercolor)
    citecolor=darkgreen,        % color of links to bibliography
    filecolor=magenta,      % color of file links
    urlcolor=cyan           % color of external links
}
\algnewcommand\algorithmicswitch{\textbf{switch}}
\algnewcommand\algorithmiccase{\textbf{case}}

\algdef{SE}[SWITCH]{Switch}{EndSwitch}[1]{\algorithmicswitch\ #1\ \algorithmicdo}{\algorithmicend\ \algorithmicswitch}%
\algdef{SE}[CASE]{Case}{EndCase}[1]{\algorithmiccase\ #1}{\algorithmicend\ \algorithmiccase}%
\algtext*{EndSwitch}%
\algtext*{EndCase}%
\newcommand{\Oish}{\widetilde{O}}
\newcommand{\eps}{\varepsilon}

\newcommand{\dist}{\mbox{\tt dist}}
\renewcommand{\paragraph}[1]{\vspace{0.15cm}\noindent {\bf #1}:}

\newcommand{\SWFT}{\mathsf{ComputeSourcewiseFT}\xspace}
\newcommand{\tree}{\mathcal{T}}

\newcommand{\FullOrShort}{full}
\def\DP{\mbox{\tt DP}}

\def\LastE{\mbox{\tt LastE}}

\ifthenelse{\equal{\FullOrShort}{full}}{
	
  \newcommand{\fullOnly}[1]{#1}
  \newcommand{\shortOnly}[1]{}
  
  }{

    \newcommand{\fullOnly}[1]{}
    \newcommand{\IncludePictures}[1]{}
   
  }

\usepackage{color}

\ifdefined\DEBUG

\newcommand{\fab}[1]{\textcolor{red}{#1}}

  \def\rem#1{{\marginpar{\raggedright\scriptsize #1}}}
   \newcommand{\fabr}[1]{\rem{\textcolor{red}{$\bullet$ #1}}}
   \newcommand{\merr}[1]{\rem{\textcolor{blue}{$\bullet$ #1}}}
   \newcommand{\grer}[1]{\rem{\textcolor{green}{$\bullet$ #1}}}
   \newcommand{\virr}[1]{\rem{\textcolor{cyan}{$\bullet$ #1}}}  
  
\else

  \newcommand{\fab}[1]{#1}

  \newcommand{\fabr}[1]{}
  \newcommand{\merr}[1]{}
  \newcommand{\grer}[1]{}
  \newcommand{\virr}[1]{}
\fi

\begin{document}

\date{}

\title{Preserving Distances in Very Faulty Graphs}
%\\\grer{I'm not sure I love this title, it's just a placeholder.  I'm interested in your suggestions.} \mer{Option 2: How to Compress Distances in a Graph in the Presence of Faults? Algorithms and Impossibilities Results \\ Option 3: Graph Distance Compression in the Presence of Faults: Breaking the $f\geq 2$ Faults Barrier}}
\author{
 Greg Bodwin\\
Stanford, gbodwin@stanford.edu
\and
 Fabrizio Grandoni\\
  IDSIA, USI-SUPSI, fabrizio@idsia.ch
\and
Merav Parter\\
	CSAIL, MIT, parter@mit.edu
		\and
Virginia Vassilevska Williams \\
	Stanford, virgi@cs.stanford.edu	
}

%\author{
% Greg Bodwin\\
%\small Stanford \\
%  \small gbodwin@stanford.edu
%\and
% Fabrizio Grandoni\\
%  \small IDSIA, USI-SUPSI \\
%  \small fabrizio@idsia.ch
%\and
%Merav Parter\\
%	\small CSAIL, MIT \\
%	\small parter@mit.edu
%		\and
%Virginia Vassilevska Williams \\
%	\small Stanford \\
%	\small virgi@cs.stanford.edu	
%}
%

\maketitle

\thispagestyle{empty}
\setcounter{page}{0}

\begin{abstract}
Preservers and additive spanners are sparse (hence cheap to store) subgraphs that preserve the distances between given pairs of nodes
exactly or with some small additive error, respectively. Since real-world networks are prone to failures, it makes sense to study fault-tolerant versions of the above structures. This turns out to be a surprisingly difficult task. 
For every small but arbitrary set of edge or vertex failures, the preservers and spanners need to contain {\em replacement paths} around the faulted set. Unfortunately, 
the complexity of the interaction between replacement paths blows up significantly, even from $1$ to $2$ faults, and the structure of optimal preservers and spanners is poorly understood. In particular, no nontrivial bounds for preservers and additive spanners are known when the number of faults is bigger than $2$.

Even the answer to the following innocent question is completely unknown: what is the worst-case size of a preserver for a \emph{single pair} of nodes in the presence of $f$ edge faults? There are no super-linear lower bounds, nor subquadratic upper bounds for $f>2$. In this paper we make substantial progress on this and other fundamental questions:

\begin{itemize}
\item[$\bullet$] We present the first truly sub-quadratic size single-pair preservers in unweighted (possibly directed) graphs for \emph{any} fixed number $f$ of faults. Our result indeed generalizes to the single-source case, and can be used to build new fault-tolerant additive spanners (for all pairs). 

\item[$\bullet$] The size of the above single-pair preservers is $O(n^{2-g(f)})$ for some positive function $g$, and grows to $O(n^2)$ for increasing $f$. We show that this is necessary even in undirected unweighted graphs, and even if you allow for a small additive error: If you aim at size $O(n^{2-\eps})$ for $\eps>0$, then the additive error has to be $\Omega(\eps f)$. This surprisingly matches known upper bounds in the literature. 

\item[$\bullet$] For weighted graphs, we provide matching upper and lower bounds for the single pair case. Namely, the size of the preserver is $\Theta(n^2)$ for $f\geq 2$ in both directed and undirected graphs, while for $f=1$ the size is $\Theta(n)$ in undirected graphs. For directed graphs, we have a superlinear upper bound and a matching lower bound.
\end{itemize}
Most of our lower bounds extend to the distance oracle setting, where rather than a subgraph we ask for any compact data structure.
%\fabr{Maybe we want to say a little more about this, not sure}  
\end{abstract}

%\newpage

%\input{Intro}
%\input{Notation}

\section{Introduction}
Distance preservers and additive spanners are (sparse) subgraphs that preserve, either exactly or with some small additive error, the distances between given critical pairs $P$ of nodes. This has been a subject of intense research in the last two decades \cite{celkin,bv-soda16,althofer,aingworth,Chechik13,BaKaMePe05,AmirGreg,Pettie09}. 

However, real-world networks are prone to failures. For this reason, more recently (e.g. \cite{ChechikLPR09,BraunschvigCPS15,chechik2010rigid,PPFTBFS13,parter2014vertex,bilo2014fault,parter2014fault,bilo2015improved,dinitz2011fault,levcopoulos2002improved,czumaj2004fault,lukovszki1999new}) researchers have devoted their attention to fault-tolerant versions of the above structures, where distances are (approximately) preserved also in the presence of a few edge (or vertex) faults. For the sake of simplicity we focus here on edge faults, but many results generalize to the case of vertex faults where $F\subseteq V$. 
\begin{definition}
Given an $n$-node graph $G=(V,E)$ and $P\subseteq V\times V$, a subgraph $H\subseteq G$ is an $f$-fault tolerant ($f$-FT) $\beta$-additive $P$-pairwise spanner if
$$
dist_{H\setminus F}(s,t)\leq dist_{G\setminus F}(s,t)+\beta,\quad\forall (s,t)\in P, \forall F\subseteq E, |F|\leq f.
$$ 
If $\beta=0$, then $H$ is an $f$-FT $P$-pairwise preserver.
\end{definition}
Finding sparse FT spanners/preservers turned out to be an incredibly challenging task. Despite intensive research, many simple questions have remained open, the most striking of which arguably is the following:

\begin{question}\label{que:linear} What is the worst-case size of a  preserver for a single pair $(s,t)$ and $f\geq 1$ faults?
\end{question} 
Prior work~\cite{ParterPODC15,PPFTBFS13} considered the single-source $\fab{P=}\{s\}\times V$ unweighted case, providing super-linear lower bounds for any $f$ and tight upper bounds for $f=1,2$. However, first, there is nothing known for $f>2$, and second, the lower bounds for the $\{s\}\times V$ case do not apply to
the single pair case where much sparser preservers might exist. Prior to this work, it was conceivable that in this case $O(n)$ edges suffice for arbitrary fixed $f$.

Our first result is a {\em complete} answer to Question \ref{que:linear} for weighted graphs. In more detail, we prove:
\begin{itemize}
\item[$\bullet$] An $(s,t)$ preserver in a weighted graph for $f=1$ has size $\Theta(n)$ in the undirected setting \fab{(Theorem \ref{thm:uwstub}, with extensions in Theorem \ref{thr:UBqpairs})} or $\Theta(\DP(n))$ in the directed setting \fab{(Theorem \ref{thr:UBLB1FT})}. 
\item[$\bullet$] An $(s,t)$ preserver in a weighted graph for $f\geq 2$ has size $\Theta(n^2)$ even in the undirected case \fab{(Theorem \ref{thr:LBweighted2FT})}. 
\end{itemize}
The function $\DP(n)$ above denotes a tight bound for the sparsity of a pairwise distance preserver in directed weighted graphs with $n$ nodes and $O(n)$ pairs. Coppersmith and Elkin~\cite{celkin} show that $\Omega(n^{4/3}) \leq \DP(n)\leq O(n^{3/2})$. It is a major open question to close this gap, and we show that the no-fault $n$-pair distance preserver question is equivalent to the $1$-fault single pair preserver question, thereby fully answering the latter question, up to resolving the major open problem for $n$-pair preservers.

For unweighted graphs, we achieve several non-trivial lower bounds concerning the worst-case size of $(s,t)$ preservers and spanners: 
\begin{itemize}
\item[$\bullet$] In the unweighted directed or undirected case this size is $\Theta(n)$ for $f=1$. This shows an interesting gap w.r.t. to the weighted case mentioned before.

\item[$\bullet$] The size is super-linear for any $f\geq 2$ even in unweighted undirected graphs and even if we allow a small enough polynomial additive error $n^{\delta}$.
\end{itemize}

Note that the latter lower bound (unlike in the weighted case) leaves room for improvements. In particular, consider the following question:
\begin{question}\label{que:subquadratic} In unweighted graphs, is the worst-case size of an $f$-FT $(s,t)$ preserver subquadratic for every constant $f\geq 2$?
\end{question} 
Prior work showed that the answer is YES for $f=\fab{1,}2$ \cite{ParterPODC15,parter2014fault}, but nothing is known for $f\geq 3$. We show that the answer is YES:
\begin{itemize}
\item[$\bullet$] In unweighted directed or undirected graphs, for any $f\geq 1$ there is an $(s,t)$ preserver of size $O(n^{2-g(f)})$ for some positive decreasing function $g(\cdot)$. \fab{See Theorem \ref{thm:multisourcef}.}
\end{itemize}

The above result has many strengths.
First, it extends to the single-source case (i.e., $P=\{s\}\times V$). 
Second, the same result holds for {\em any} fixed number $f$ of vertex faults. Prior work was only able to address the simple case $f=1$~\cite{parter2014vertex}.
Third, such a preserver can be computed very efficiently in $O(fmn)$ time, and its analysis is relatively simple (e.g., compared to the slightly better size bound in \cite{ParterPODC15} that was achieved by a cumbersome case analysis).
Finally, via fairly standard techniques, the preserver result also implies improved $f$-FT $2$-additive (all pairs!) spanners for all $f\geq 1$ \fab{(see Theorem \ref{thr:additivespanner})}.
%\fabr{This is how much I want to say about spanners here}

In the above result the size of the preserver grows quickly to $O(n^2$) for increasing $f$. This raises the following new question:
\begin{question}\label{que:strongly} Does there exist a universal constant $\eps>0$ such that all unweighted graphs 
have an $f$-FT $(s,t)$ preserver of size $O_f(n^{2-\eps})$? What if we allow a small additive error? 
\end{question}
The only result with \emph{strongly sub-quadratic} size in the above sense is an $O(f\cdot n^{4/3})$ size spanner with additive error $\Theta(f)$ \cite{BraunschvigCPS15,bilo2015improved}. Can we remove or reduce the dependence of the error on $f$? We show that the answer is NO:
\begin{itemize}
\item[$\bullet$] In undirected unweighted graphs, any single-pair spanner of strongly subquadratic size $O_{\fab{f}}(n^{2-\eps})$ for $\eps>0$ needs to have additive error $\Omega(\eps f)$. (\fab{See Theorem \ref{thr:LBadditive} and related results in Theorems \ref{thm:subs-lb-improvement}-\ref{thm:gap-pair-lb}).}
%\item[$\bullet$] Even if we allow for some polynomial additive error $n^{\delta}$, the size of the spanner has to be super-linear already for $f=2$. 
\end{itemize}
Hence the linear dependence in $f$ in the additive error in \cite{BraunschvigCPS15,bilo2015improved} is indeed necessary.  We found this very surprising.
The table in Appendix \ref{sec:tabFig} summarizes our main results for FT-preservers. 

So far we have focused on sparse distance preserving {\em subgraphs}. However, suppose that the distance estimates can be stored in a different way in memory. Data structures that store the distance information of a graph in the presence of faults are called {\em distance sensitivity oracles}. Distance sensitivity oracles are also intensely studied \cite{demetrescu2008oracles,BernsteinK09,weimann2013replacement,GrandoniW12,duan2009dual,DP17}.
Our main goal here is to keep the size of the data structure as small as possible. Other typical goals are to minimize preprocessing and query time - we will not address these.
\begin{question}\label{que:oracle} How much space do we need to preserve (exactly or with a small additive error) the distances between a given pair of nodes in the presence of $f$ faults?
\end{question} 
Clearly all our preserver/spanner upper bounds extend to the oracle case, however the lower bounds might not: in principle a distance oracle can use much less space than a preserver/spanner with the same accuracy. Our main contribution here is the following incompressibility result:\fabr{I put the claim back}
\begin{itemize}
\item[$\bullet$] The worst-case size of a single-pair exact distance sensitivity oracle in directed or undirected weighted graphs is $\Theta(n^2)$ for $f \geq  2$ (note that the optimal size for $f=1$ is $\Theta(n)$ by simple folklore arguments, so our result completes these settings). See Theorem \ref{thr:LBweighted2FT}.
\item[$\bullet$] If we allow for a polynomial additive error $n^{\delta}$, for small $\delta$, even in the setting of undirected unweighted graphs, then the size of the oracle has to be super-linear already for $f\geq 3$ \fab{(Theorem \ref{thm:gap-compression-lb})}.
\end{itemize}

The technical part of the paper has precise theorem statements for all results.
The interested reader will find even more results and corollaries there as well. We omitted these from this introduction for the sake of clarity.

%\fabr{Recall to put exact theorem claims in technical part since we don't have them in intro now}

\subsection{Related Work}
%\fabr{Related work was too much detailed}
Fault-tolerant spanners were introduced in the geometric setting \cite{levcopoulos2002improved} (see also \cite{lukovszki1999new,czumaj2004fault}). FT-spanners with multiplicative stretch are relatively well understood: the error/sparsity for $f$-FT and $f$-VFT multiplicative spanners is (up to a small polynomial factor in $f$) the same as in the nonfaulty case. For $f$ edge faults, 
Chechik et al.~\cite{ChechikLPR09}
 showed how to construct $f$-FT $(2k-1)$-multiplicative spanners with size $\tilde{O}(fn^{1+\frac{1}{k}})$ for any $f, k \geq 1$. 
They also construct an $f$-VFT spanner with the same stretch and larger size.
This was later improved by Dinitz and Krauthgamer~\cite{dinitz2011fault} who showed the construction of $f$-VFT spanners with $2k-1$ error and $\tilde{O}\left(f^{2-\frac{1}{k}}n^{1+\frac{1}{k}}\right)$ edges.

FT additive spanners were first considered by Braunschvig, Chechik and Peleg in \cite{BraunschvigCPS15} (see also \cite{bilo2015improved} for slightly improved results). They showed that FT $\Theta(f)$-additive spanners can be constructed by combining FT multiplicative spanners with (non-faulty) additive spanners. This construction, however, supports only edge faults. Parter and Peleg showed in \cite{parter2014fault} a lower bound of $\Omega(n^{1+\eps_{\beta}})$ edges for single-source FT $\beta$-additive spanners. They also provided a construction of single-source FT-spanner with additive stretch $4$ and $O(n^{4/3})$ edges that is resilient to one edge fault. The first constructions of FT-additive spanners resilient against \emph{one vertex} fault were given in \cite{parter2014vertex} and later on in \cite{bilo2015improved}. Prior to our work, no construction of FT-additive spanners was known for $f\geq 2$ vertex faults. 

As mentioned earlier, the computation of preservers and spanners in the non-faulty case (i.e. when $f=0$) has been the subject of intense research in the last few decades. The current-best preservers can be found in \cite{celkin, bv-soda16, Bodwin17}. Spanners are also well understood, both for multiplicative stretch \cite{althofer,erdos} and for additive stretch \cite{aingworth,Chechik13,BaKaMePe05,Woodruff10,AmirGreg,bv-soda16,Chechik13,Pettie09,ABP17}. There are also a few results on ``mixed'' spanners with both multiplicative and additive stretch \cite{ElkinP04,ThorupZ06,BaKaMePe05}

Distance sensitivity oracles are data structures that can answer queries about the distances in a given graph in the presence of faults. The first nontrivial construction was given by Demetrescu et al. \cite{demetrescu2008oracles} and later improved by Bernstein and Karger~\cite{BernsteinK09} who showed how to construct $\tilde{O}(n^2)$-space, constant query time oracles for a single edge fault for an $m$-edge $n$-node graph in $\tilde{O}(mn)$ time. The first work that considered the case of two faults (hence making the first jump from one to two) is due to Duan and Pettie in \cite{duan2009dual}. Their distance oracle has nearly optimal size of $\widetilde{O}(n^2)$ and query time of $\widetilde{O}(1)$. The case of bounded edge weights, and possibly multiple faults, is addressed in \cite{weimann2013replacement,GrandoniW12} exploiting fast matrix multiplication techniques. The size of their oracle is super-quadratic.

The notion of FT-preservers is also closely related to the problem of
constructing \emph{replacement paths}. For a pair of vertices $s$ and $t$ and an edge $e$, the replacement path $P_{s,t,e}$ is the $s$-$t$ shortest-path that avoids $e$\footnote{Replacement paths were originally defined for the single edge fault case, but later on extended to the case of multiple faults as well.}. The efficient computation of replacement paths is addressed, among others, in \cite{malik1989k,roditty2012replacement,weimann2013replacement,williams2011faster}. A single-source version of the problem is studied in \cite{GrandoniW12}.
Single-source FT structures that preserve strong connectivity have been studied in \cite{BaswanaCR16}.

\subsection{Preliminaries and Notation}

Assume throughout that all shortest paths ties are broken in a consistent manner. For every $s,t\in V$ and a subgraph $G'\subseteq G$, let $\pi_{G'}(s,t)$ be the (unique) $u$-$v$ shortest path in $G'$ (i.e., it is unique under breaking ties). If there is no path between $s$ and $t$ in $G'$, we define $\pi_{G'}(s,t)=\emptyset$.  When $G'=G$, we simply write $\pi(u,v)$.  For any path $P$ containing nodes $u,v$, let $P[u \leadsto v]$ be the subpath of $P$ between $u$ and $v$. 
%Given an $s$-$t$ path $P$ and an edge $e=(x,y) \in P$, let $\dist_{P}(s, e)$ be the distance  between $s$ and the closest node in $\{x,y\}$ along $P$.
For $s,t \in V$ and $F \subseteq E$, we let $P_{s,t,F}=\pi_{G\setminus F}(s,t)$ be the $s$-$t$ shortest-path in $G \setminus F$. We call such paths \emph{replacement paths}. When $F=\{e\}$, we simply write $P_{s,t,e}$.  By $m$ we denote the number of edges in the graph currently being considered.
%\fabr{Maybe we don't need appendix related work: enough to give quick references} 
%In Appendix \ref{appen:addrelatedwork}, we review additional related work.

\medskip
The structure of the paper is as follows. In Sec. \ref{sec:UB}, we describe an efficient construction for FT-preservers and additive spanners with a subquadratic number of edges. Then, in Sec. \ref{sec:lbST}, we provide several lower bound constructions for a single $s$-$t$ pair, both for the exact and for the additive stretch case. Finally, in Sec. \ref{sec:Weighted} we consider the setting of weighted graphs. Most of the results of that setting are deferred to Appendix \ref{sec:apxWeighted}.  Missing proofs in other sections can be found in the appendix as well.

\section{Efficient Construction of FT-Preservers and Spanners}
\label{sec:UB}
In this section we show:
\begin{theorem}
\label{thm:multisourcef}
For every directed or undirected unweighted graph $G=(V,E)$, integer $f \geq 1$ and $S \subseteq V$, 
one can construct in time $O(f \, n \, m)$ an $f$-FT $S$-sourcewise (i.e. $P = S \times V$) preserver of size $\widetilde{O}(f\cdot |S|^{1/2^f} \cdot n^{2-1/2^f})$. 
\end{theorem}
We remark that Theorem \ref{thm:multisourcef} holds under both edge and vertex faults. 
We next focus on the directed case, the undirected one being analogous and simpler. We begin by recapping the currently-known approaches for handling many faults, and we explain why these approaches fail to achieve interesting space/construction time bounds for large $f$.

%\fabr{I shortened the initial part in order to make room to proof of Lemma 5 previously in appendix (quite crucial).}

\paragraph{The limits of previous approaches} A known approach for handling many faults is by \emph{random sampling} of subgraphs, as  introduced by Weimann and Yuster \cite{weimann2013replacement} in the setting of distance sensitivity oracles, and later on applied by Dinitz and Kraughgamer \cite{dinitz2011fault} in the setting of fault tolerant spanners. The high level idea is to generate multiple subgraphs $G_1, \ldots, G_r$ by removing each edge/vertex independently with sufficiently large probability $p$; intuitively, each $G_i$ simultaneously captures many possible fault sets of size $f$. One can show that, for a sufficiently small parameter $L$ and for any given (\emph{short}) replacement path $P_{s,t,F}$ of length at most $L$ (avoiding faults $F$), w.h.p. in at least one $G_i$ the path $P_{s,t,F}$ is still present while all edges/vertices in $F$ are deleted. Thus, if we compute a (non-faulty) preserver $H_i \subseteq G_i$ for each $i$, then the graph $H = \bigcup_i H_i$ will contain every short replacement path. For the remaining (\emph{long}) replacement paths, Weimann and Yuster use a random decomposition into short subpaths. Unfortunately, any combination of the parameters $p,r,L$ leads to a quadratic (or larger) space usage. 

%\fabr{I put back this discussion from appendix since it seems useful}
Another way to handle multiple faults is by extending the approach in \cite{PPFTBFS13,parter2014fault,parter2014vertex} that works for $f\in \{1,2\}$. A useful trick used in those papers (inspired by prior work in \cite{roditty2012replacement,williams2011faster}) is as follows: suppose $f=1$, and fix a target node $t$. Consider the shortest path $\pi(s,t)$. It is sufficient to take the \emph{last} edge of each replacement path $P_{s,t,e}$ and charge it to the node $t$; the rest of the path is then charged to other nodes by an inductive argument. Hence, one only needs to bound the number of \emph{new-ending} paths -- those that end in an edge that is not already in $\pi(s,t)$.
In the case $f=1$, these new-ending paths have a nice structure: they diverge from $\pi(s,t)$ at some vertex $b$ \fab{(\emph{divergence point})} above the failing \fab{edge/}vertex and collide again with $\pi(s,t)$ only at the terminal $t$; the subpath connecting $b$ and $t$ on the replacement path is called its \emph{detour}.
One can divide the $s$-$t$ replacement paths into two groups: short (resp., long) paths are those whose detour has length at most (resp., at least) $\sqrt{n}$.
It is then straightforward enough to show that each category of path contributes only $\Oish(n^{1/2})$ edges entering $t$, and so (collecting these last edges over all nodes in the graph) the output subgraph has $\widetilde{O}(n^{3/2})$ edges in total.
Generalizing this to the case of multiple faults is non-trivial already for the case of $f=2$. The main obstacle here stems from a lack of structural understanding of replacement paths for multiple faults: in particular, any given divergence point $b \in \pi(s,t)$ can now be associated with many new-ending paths and not only one!
In the only known positive solution for $f=2$ \cite{ParterPODC15}, the approach works only for edge faults and is based on an extensive case analysis whose extension to larger $f$ is beyond reasonable reach.
Thus, in the absence of new structural understanding, further progress seems very difficult.

A second source of difficulties is related to the {\em running time} of the construction. A priori, it seems that constructing a preserver $H$ should require computing all replacement paths $P_{s,t,F}$, which leads to a construction time that scales {\em exponentially} in $f$. In particular, by deciding to omit an edge $e$ from the preserver $H$, we must somehow check that this edge does not appear on \emph{any} of the replacement paths $P_{s,t,F}$ (possibly, without computing these replacement paths explicitly).
%We answer this question in the affirmative and show how to compute a subquadratic-size preserver for $f=o(\log\log n)$ faults in time $O(f  n m)$.
%In fact, our algorithm consists simply of $O(n)$ rounds of computation of $f$ BFS trees!
%The perhaps surprising correctness proof comes from the fact that each of these $f \cdot n$ BFS trees is computed in a different graph $G' \subseteq G$. Namely, for every vertex $v$ we carefully and quite nontrivially design a set of $f$ graphs that are {\em specific} to that $v$.

\paragraph{Our basic approach} The basic idea behind our algorithm is as follows. Similar to \cite{PPFTBFS13,parter2014fault,parter2014vertex}, we focus on each target node $t$, and define a set $E_t$ of edges incident to $t$ to be added to our preserver. Intuitively, these are the last edges of new-ending paths as described before. The construction of $E_t$, however, deviates substantially from prior work. Let us focus on the simpler case of edge deletions. The set $E_t$ is constructed recursively, according to parameter $f$. Initially we consider the shortest path tree $T$ from the source set $S$ to $t$, and add to $E_t$ the edges of $T$ incident to $t$ (at most $|S|$ many). Consider any new-ending replacement path $P$ for $t$. By the previous discussion, this path has to leave $T$ at some node $b$ and it meets $T$ again only at $t$: let $D$ be the subpath of $P$ between $b$ and $t$ (the \emph{detour} of $P$). Note that $D$ is edge-disjoint from $T$, i.e. it is contained in the graph $G'=G\setminus E(T)$. Therefore, it would be sufficient to compute recursively the set $E'_t$ of final edges of new-ending replacement paths for $t$ in the graph $G'$ with source set $S'$ given by the possible \fab{divergence points} $b$ and w.r.t. $f-1$ faults (recall that one fault must be in $E(T)$, hence we avoid that anyway in $G'$). This set $E'_t$ can then be added to $E_t$.

The problem with this approach is that $S'$ can contain $\Omega(n)$ many divergence points \fab{(hence $E_t$ $\Omega(n)$ many edges)}, leading to a \fab{trivial $\Omega(n^2)$ size} preserver. In order to circumvent this problem, we classify the divergence points $b$ in two categories. Consider first the nodes $b$ at distance at most $L$ from $t$ \fab{along $T$}, for some parameter $L$. There are only $O(|S|L)$ many such nodes $S^{short}$, which is sublinear for $|S|$ and $L$ small enough. Therefore we can safely add $S^{short}$ to $S'$. For the remaining divergence points $b$, we observe that the corresponding detour $D$ must have length at least $L$: therefore by sampling $\tilde{O}(n/L)$ nodes $S^{long}$ we hit all such detours w.h.p. Suppose that $\sigma\in S^{long}$ hits detour $D$. Then the portion of $D$ from $\sigma$ to $t$ also contains the final edge of $D$ to be added to $E_{\fab{t}}$. In other terms, it is sufficient to add $S^{long}$ (which has sublinear size for polynomially large $L$) to $S'$ to cover all the detours of nodes $b$ of the second type. Altogether, in the recursive call we need to handle one \fab{less} fault w.r.t. a larger (but sublinear) set of sources $S'$. Our approach has several benefits:
\begin{itemize}\itemsep0pt
\item It leads to a subquadratic size for any $f$ (for a proper choice of the parameters);
\item It leads to a very fast algorithm. In fact, for each target $t$ we only need to compute a BFS tree in $f$ different graphs, leading to an $O(f n m)$ running time;
\item Our analysis is very simple, much simpler than in \cite{ParterPODC15} for the case $f=2$;
\item It can be easily extended to the case of vertex faults. 
\end{itemize}

\paragraph{Algorithm for Edge Faults} Let us start with the edge faults case. The algorithm constructs a set $E_t$ of edges incident to each target node $t\in V$. The final preserver is simply the union $H= \bigcup_{t\in V} E_t$ of these edges. We next describe the construction of each $E_t$ (see also Alg. \ref{FTalgorithm}). The computation proceeds in rounds $i=0,\ldots, f$. At the beginning of round $i$ we are given a subgraph $G_i$ (with $G_0=G$) and a set of sources $S_i$ (with $S_0=S$). 

We compute a partial  
BFS tree $T_i=\bigcup_{s \in S_i} \pi_{G_i}(s,t)$\footnote{If $\pi_{G_i}(s,t)$ does not exist, recall that we define it as an empty set of edges.} from $S_i$ to $t$, and add to $E_t$ (which is initially empty) the edges $\{\LastE(\pi_{T_i}(s,t)) ~\mid~ s \in S_i\}$ of this tree incident to $t$. Here, for a path $\pi$ where one endpoint is the considered target node $t$, we denote by $\LastE(\pi)$ the edge of $\pi$ incident to $t$.
The source set $S_{i+1}$ is given by $S^{short}_i\cup S^{long}_i$. Here $S^{short}_i=\{ v \in V(T_i) ~\mid~ \dist_{T_i}(v,t) \leq d_i\}$ is the set of nodes at distance at most $d_i =\sqrt{n/|S_i|\cdot f \log n}$ from $t$, while $S^{long}_i$ is a random sample of $\Theta(n/d_i\cdot f \log n)$ vertices. 
The graph $G_{i+1}$ is obtained from $G_i$ be removing the edges $E(T_i)$\footnote{Note that for $f=1$, the algorithm has some similarity to the replacement path computation of \cite{roditty2012replacement}. Yet, there was no prior extension of this idea for $f\geq 2$.}.

\begin{algorithm}[t]
\caption{Construction of $E_t$ in our $f$-FT $S$-Sourcewise Preserver Algorithm.}
\label{FTalgorithm}
\begin{algorithmic}[1]
\Procedure{$\SWFT(t,S,f,G)$} {}
\begin{flushleft}
Input: A graph $G$ with a source set $S$ and terminal $t$, number of faults $f$. \\
Output: Edges $E_t$ incident to $t$ in an $f$-FT $S$-sourcewise preserver $H$.
\end{flushleft}
\\
Set $G_0=G$, $S_0=S$, $E_t=\emptyset$.
\For{$i \in \{0,\ldots, f\}$}
\State Compute the partial BFS tree $T_i=\bigcup_{s \in S_i} \pi_{G_i}(s,t)$.  
\State $E_t = E_t \cup \{\LastE(\pi_{T_i}(s,t)) ~\mid~ s \in S_i\}$.
\State Set distance threshold $d_i =\sqrt{n/|S_i|\cdot f \log n}$.
\State Let $S^{short}_i=\{ v \in V(T_i) ~\mid~ \dist_{T_i}(v,t) \leq d_i\}$. 
\State Sample a collection $S^{long}_i \subseteq V(G_i)$ of $\Theta(n/d_i\cdot f \log n)$ vertices.
\State Set $S_{i+1}=S^{short}_i \cup S^{long}_i$ and $G_{i+1}=G_{i} \setminus E(T_i)$.  
\EndFor
\EndProcedure
\end{algorithmic}
\end{algorithm}

\paragraph{Adaptation for Vertex Faults}
The only change in the algorithm is in the definition of the graph $G_i$ inside the procedure to compute $E_t$. We cannot allow ourselves to remove all the vertices of the tree $T_i$ from $G_i$ and hence a more subtle definition is required.
To define $G_{i+1}$, we first remove from $G_i$: (1) all edges of $S^{short}_i \times S^{short}_i$, (2)  the edges of $E(T_i)$, and (3) the vertices of $V(T_i) \setminus S^{short}_i$. In addition, we orient all remaining edges incident to $S^{short}_i$ to be directed \emph{away} from these vertices (i.e., the incoming degree of the $S^{short}_i$ vertices in $G_{i+1}$ is zero). Finally, we delete all remaining edges incident to $S^{short}_i$ which are directed towards any one of these vertices (i.e., the incoming degree of the $S^{short}_i$ vertices in $G_{i+1}$ is zero).

%In addition, we orient all remaining edges incident to $S^{short}_i$ to be directed \emph{away} from these vertices (i.e., the incoming degree of the $S^{short}_i$ vertices in $G_{i+1}$ is zero). 

\paragraph{Analysis}
We now analyze our algorithm.  Since for each vertex $t$, we compute $f$ (partial) BFS trees, we get trivially:
\vspace{-4pt}
\begin{lemma}[Running Time]\label{lem:UPtime}
The subgraph $H$ is computed within $O(f\, n\, m)$ time. 
\end{lemma}
\vspace{-4pt}
We proceed with bounding the size of $H$.
\vspace{-4pt}
\begin{lemma}[Size Analysis]\label{lem:UPsize}
$|E_t|=\widetilde{O}(|S|^{1/2^f} \cdot (fn)^{1-1/2^f})$ for every $t \in V$, hence $|E(H)|=\widetilde{O}(f|S|^{1/2^f} n^{2-1/2^f})$. 
\end{lemma}

\begin{proof}%[Proof of Theorem \ref{lem:UPsize}]
Since the number of edges collected at the end each round $i$ is bounded by the number of sources $S_i$, it is sufficient to bound $|S_i|$ for all $i$. Observe that, for every $i  \in \{0,\ldots, f-1\}$, 
$$
|S_{i+1}|\leq |S^{long}_i|+|S^{short}_i| \leq d_i\cdot |S_i|+\Theta(n/d_i\cdot f \log n) =\Theta(d_i\cdot |S_i|).
$$ 
By resolving this recurrence starting with $|S_0|=|S|$ one obtains $|S_i|=O(|S|^{1/2^i}(fn\log n)^{1-1/2^i})$. The claim follows by summing over $i\in \{0,\ldots,f\}$.
\end{proof}
%}\APPENDSIZE

%\paragraph{Correctness Analysis (for vertex faults)}
We next show that the algorithm is correct. We focus on the vertex fault case, the edge fault case being similar and simpler.
Let us define, for $t\in V$ and $i \in \{0,\ldots, f\}$,
$$
\mathcal{P}_{t,i}=\{\pi_{G_i \setminus F}(s,t) ~\mid~ s\in S_i, ~F \subseteq V(G_i),~ |F|\leq f-i\}.
$$
\vspace{-20pt}
\begin{lemma}
\label{lem:maininduc}
For every $t\in V$ and $i \in \{0,\ldots, f\}$, it holds that:
$\LastE(\pi) \in E_t \mbox{~~for every ~~} \pi \in  \mathcal{P}_{t,i}.$
\end{lemma}
%\vspace{-3pt}
%\def\APPENDINDUC{
\begin{proof}%[Proof of Lemma \ref{lem:maininduc}]
We prove the claim by decreasing induction on $i \in \{f, \ldots, 0\}$. For the base of the induction, consider the case of $i=f$. In this case, $\mathcal{P}_{t,f}=\{\pi_{G_f}(s,t) ~\mid~ s \in S_f\}$.  Since we add precisely the last edges of these paths to the set $E_t$, the claim holds.
Assume that the lemma holds for rounds $f, f-1, \ldots, i+1$ and consider round $i$. 
%We will prove using the induction assumption, that $\LastE(P) \in E_t$ for every $P \in \mathcal{P}_{t,i}$. 
For every $\pi_{G_i \setminus F}(s,t) \in \mathcal{P}_{t,i}$, let $P'_{s,t,F}=\pi_{G_i \setminus F}(s,t)$. \footnote{We denote these replacement paths as $P'_{s,t,F}$ as they are computed in $G_i$ and not in $G$.} 
Consider the partial BFS tree $T_i=\bigcup_{s \in S_i}\pi_{G_i}(s,t)$ rooted at $t$. 
Note that all (interesting) replacement paths $P'_{s,t,F}$ in $G_i$ have at least one failing vertex $v \in F \cap V(T_i)$ as otherwise $P'_{s,t,F}=\pi_{G_i}(s,t)$. 

We next partition the replacement paths $\pi\in\mathcal{P}_{t,i}$ into two types depending on their last edge $\LastE(\pi)$. The first class contains all paths whose last edge is in $T_i$. The second class of replacement paths contains the remaining paths, which end with an edge that is not in $T_i$. We call this second class of paths \emph{new-ending} replacement paths. Observe that the first class is taken care of, since we add all edges incident to $t$ in $T_i$. Hence it remains to prove the lemma for the set of new-ending paths.

For every new-ending path $P'_{s,t,F}$, let $b_{s,t,F}$ be the last vertex on $P'_{s,t,F}$ that is in $V(T_i)\setminus \{t\}$. We call the vertex $b_{s,t,F}$ the \emph{last divergence point} of the new-ending replacement path. Note that the detour $D_{s,t,F}=P'_{s,t,F}[b_{s,t,F}\leadsto t]$ is vertex disjoint with the tree $T_i$ except for the vertices $b_{s,t,F}$ and $t$. From now on, since we only wish to collect last edges, we may restrict our attention to this detour subpath. That is, since $\LastE(D_{s,t,F})=\LastE(P'_{s,t,F})$, it is sufficient to show that $\LastE(D_{s,t,F}) \in E_t$.

Our approach is based on dividing the set of new-ending paths in $\mathcal{P}_{t,i}$ into two classes based on the position of their last divergence point $b_{s,t,F}$ (see Fig. \ref{fig:ftmulti}). The first class $\mathcal{P}_{short}$ consists of new-ending paths in $\mathcal{P}_{t,i}$ whose last divergence point is at distance at most $d_i=\sqrt{n /|S_i| \cdot f\log n}$ from $t$ on $T_i$. In other words, this class contains all new-ending paths whose last divergence point is in the set $S^{short}_i$. 
We now claim the following.
\begin{figure}[h]
\begin{center}
\includegraphics[scale=0.2]{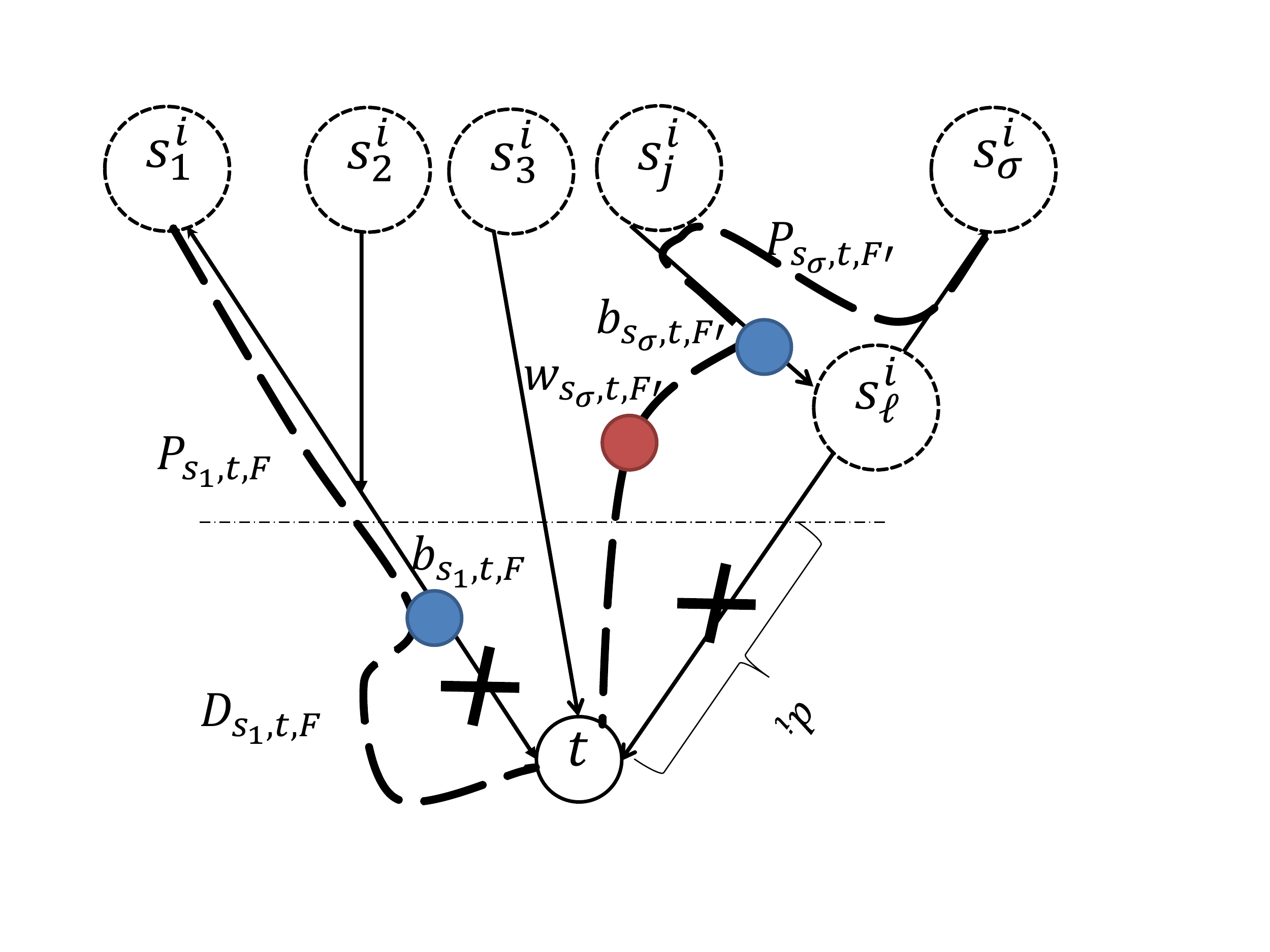}
\caption{
\label{fig:ftmulti} Shown is a partial tree $T_i$ whose leaves set are in $S_i$, all edges are directed towards $t$ in the directed case. (In the figure, we let $s^i_j=s_i$ for simplicity of notation). The replacement paths $P_{s^i_j,t,F}$ are divided into two types depending on their last divergence point $b_{s^i_j,t,F}$. Note that this point is not necessarily on $\pi(s^i_j,t)$ and may appear on other $\pi(s^i_\ell)$ paths. The vertices appearing on the first $d_i$ levels of $T_i$ are $S^{short}_i$. The path $P_{s_1,t,F}$ is in $\mathcal{P}_{short}$ and the path $P_{s_\sigma,t,F'}$ is in $\mathcal{P}_{long}$. The vertex $w_{s_\sigma,t,F'}$ is in the set $S^{long}_i$ and it hits the long detour of $P_{s_\sigma,t,F'}$. Note that since both $P_{s_1,t,F}$ and $P_{s_\sigma,t,F}$ are new-ending, one of the vertices in their failing set $F,F'$ appears on $\pi_{T_i}(b_{s_1,t,F},t),\pi_{T_i}(b_{s_\sigma,t,F'},t)$ respectively.}
\end{center}
\end{figure}
\begin{claim}
\label{cl:short}
For every $P'_{s,t,F} \in \mathcal{P}_{short}$, the detour $D_{s,t,F}$ is in $\mathcal{P}_{t,i+1}$. 
\end{claim}
\begin{proof}
Since $D_{s,t,F}$ is a subpath of the replacement path $P'_{s,t,F}$, $D_{s,t,f}$ is the shortest path between $b_{s,t,F}$ and $t$ in $G_i \setminus F$. Recall that $D_{s,t,F}$ is vertex disjoint with $V(T_i) \setminus \{b_{s,t,F},t\}$. 

Since $b_{s,t,F}$ is the last divergence point of $P'_{s,t,F}$ with $T_i$, the detour $D_{s,t,F}$ starts from a vertex $b_{s,t,F} \in S^{short}_i$ and does not pass through any other vertex in $V(T_i)\setminus \{t\}$. Since we only changed in $G_{i+1}$ the direction of edges incident to $S^{short}_i$ vertices but the outgoing edge connecting $b_{s,t,F}$ to its neighbor $x$ on $D_{s,t,F}[b_{s,t,F} \leadsto t]$ remains (i.e., this vertex $x$ is not in $V(T_i) \setminus \{t\}$), this implies that the detour $D_{s,t,F}$ exists in $G_{i+1}$. In particular, note that the vertex $b_{s,t,F}$ cannot be a neighbor of $t$ in $T_i$. 
If $(b_{s,t,F},t)$ were an edge in $T_i$, then we can replace the portion of the detour path between $b_{s,t,F}$ and $t$ by this edge, getting a contradiction to the fact that $P'_{s,t,F}$ is a new-ending path\footnote{For the edge fault case, the argument is much simpler: by removing $E(T_i)$ from $G_i$, we avoid at least one the failing edges in $G_{i+1}$.}.

Next, observe that at least one of the failing vertices in $F$  occurs on the subpath $\pi_{G_i}[b_{s,t,F},t]$, let this vertex be $v \in F$. Since $v \in S^{short}_i$, all the edges are directed away from $v$ in $G_{i+1}$ and hence the paths going out from the source $b_{s,t,F}$ in $G_{i+1}$ cannot pass through $v$. Letting $F'=F \setminus V(T_i)$, it holds that (1) $|F'| \leq f-i-1$ and (2) since the shortest path ties are decided in a consistent manner and by definition of $G_{i+1}$, it holds that $D_{s,t,F}=\pi_{G_{i+1} \setminus F'}(b_{s,t,F}, t)$. As $b_{s,t,F} \in S^{short}_i$, it holds that  $D_{s,t,F} \in \mathcal{P}_{t,i+1}$.
\end{proof}
Hence by the inductive hypothesis for $i+1$, $\LastE(P'_{s,t,F})$ is in $E_t$ for every $P'_{s,t,F} \in \mathcal{P}_{short}$. 
We now turn to consider the second class of paths 
$\mathcal{P}_{long}$ which contains all remaining new-ending paths; i.e., those paths whose last divergence point is at distance at least $d_i$ from $t$ on $T_i$. Note that the detour $D_{s,t,F}=P'_{s,t,F}[b_{s,t,F} \leadsto t]$ of these paths is long -- i.e., \fab{its} length is at least $d_i$. For convenience, we will consider the internal part $D'_{s,t,F}=D_{s,t,F} \setminus \{b_{s,t,F},t\}$ of these detours, so that the first and last vertices of these detours are not on $T_i$.

We exploit the lengths of these detours $D'_{s,t,F}$ and claim that 
 w.h.p, the set $S^{long}_i$ is a hitting set for these detours. This indeed holds by simple 
union bound overall possible $O(n^{f+2})$ detours. For every $P'_{s,t,F} \in \mathcal{P}_{long}$, let $w_{s,t,F} \in V(D'_{s,t,F})\cap S^{long}_i$. (By the hitting set property, w.h.p., $w_{s,t,F}$ is well defined for each long detour).
Let $W_{s,t,F}=P'_{s,t,F}[w_{s,t,F}, t]$ be the suffix of the path $P'_{s,t,F}$ starting at a vertex from the hitting set $w_{s,t,F} \in S^{long}_i$. Since $\LastE(P'_{s,t,F})=\LastE(W_{s,t,F})$, it is sufficient to show that $\LastE(W_{s,t,F})$ is in $E_t$. 
\begin{claim}
\label{cl:long} 
For every $P'_{s,t,F} \in \mathcal{P}_{long}$, it holds that $W_{s,t,F} \in \mathcal{P}_{t,i+1}$. 
\end{claim}
\begin{proof}
Clearly, $W_{s,t,f}$ is the shortest path between $w_{s,t,F}$ and $t$ in $G_i \setminus F$. Since $W_{s,t,F} \subseteq D'_{s,t,F}$ is vertex disjoint with $V(T_i)$, it holds that $W_{s,t,F}=\pi_{G_{i+1} \setminus F'}(w_{s,t,F}, t)$ for $F'=F \setminus V(T_i)$. Note that since at least one fault occurred on $T_i$, we have that $|F'|\leq f-i-1$. 
As $w_{s,t,F} \in S^{long}_i$, it holds that 
$W_{s,t,F} \in \mathcal{P}_{t,i+1}$. The lemma follows.
\end{proof}
By applying the claim for $i=0$, we get that $\LastE(P'_{s,t,F})$ is in $E_t$ as required for every $P'_{s,t,F} \in \mathcal{P}_{long}$. This completes the proof.
\end{proof}
%}%\APPENDINDUC
%We are now ready to complete the correctness proof.
\begin{lemma} (Correctness)\label{lem:UBcorrect}
$H$ is an $f$-FT $S$-sourcewise preserver.
\end{lemma}
\begin{proof}
By using Lemma \ref{lem:maininduc} with $i=0$, we get that for every $t \in V$, $s \in S$ and $F \subseteq V$, $|F|\leq f$, $\LastE(P_{s,t,F}) \in E_t$ (and hence also $\LastE(P_{s,t,F}) \in H$).
It remains to show that taking the last edge of each replacement path $P_{s,t,F}$ is sufficient. The base case is for paths of length $1$, where we have clearly kept the entire path in our preserver. Then, assuming the hypothesis holds for paths up to length $k-1$, consider a path $P_{s,t,F}$ of length $k$. Let $\LastE(P_{s,t,F})=(u,t)$. Then since we break ties in a consistent manner, $P_{s,t,F}=P_{s,u,F} \circ \LastE(P_{s,t,e})$. By the inductive hypothesis $P_{s,u,F}$ is in $H$, and since we included the last edge, $P_{s,t,F}$ is also in $H$. The claim follows.
\end{proof}
Theorem \ref{thm:multisourcef} now immediately follows from Lemmas \ref{lem:UPtime}, \ref{lem:UPsize}, and \ref{lem:UBcorrect}. Combing our $f$-FT sourcewise preserver from Theorem \ref{thm:multisourcef} with standard techniques (see, e.g. \cite{parter2014vertex}), we show:
\begin{theorem}
\label{thr:additivespanner}
For every undirected unweighted graph $G=(V,E)$
and integer $f\geq 1$, there exists a randomized $\Oish(fnm)$-time construction of a $+2$-additive $f$-FT spanner of $G$ of size $\widetilde{O}(f \cdot n^{2-1/(2^{f}+1)})$ that succeeds w.h.p.\footnote{The term w.h.p. (with high probability) here indicates a probability exceeding $1-{1}/{n^c}$, for an arbitrary constant $c\geq 2$. Since randomization is only used to select hitting sets, the algorithm can be derandomized; details will be given in the journal version.}.
\end{theorem}
\begin{proof} 
The spanner construction works as follows. Let $L$ be an integer parameter to be fixed later. A vertex $u$ is \emph{low-degree} if it has degree less than $L$, otherwise it is \emph{high-degree}. 
Let $S$ be a random sample of $\Theta(\frac{n}{L}\cdot f \log n)$ vertices. Our spanner $H$ consists of the $f$-VFT $S$-sourcewise preserver from Theorem \ref{thm:multisourcef} plus all the edges incident to low-degree vertices. We now analyze the construction.

The size of $H$ is bounded by:
$$
\widetilde{O}\left(f|S|^{1/2^f}\cdot n^{2-1/2^f}\right)+O(n L)=\widetilde{O}\left(f^{1+1/2^f} L^{-1/2^f}\cdot n^{2}+ nL\right)
$$
The claim on the size follows by choosing $L=\left\lceil f n^{2^f/(2^f+1)}\right\rceil$.
%. By Theorem \ref{thm:multisourcef}, $H_f(S)$ contains $\widetilde{O}(\fab{f}|S|^{1/2^f}\cdot n^{2-1/2^f})$ edges: by plugging the size of $|S|$, the size bound follows.

Next, we turn to show correctness. First note that w.h.p every high-degree vertex has at least $f+1$ neighbors in $S$. Consider any pair of vertices $u,t$ and a set of failing vertices $F$ and let $P_{u,t,F}$ be the $u-t$ shortest path in $G\setminus F$. Let $x$ be the last vertex (closest to $t$) incident to a missing edge $e \in P_{u,t,F}\setminus E(H)$. Hence $x$ is a high-degree vertex. We observe that, w.h.p., $x$ is adjacent to at least $f+1$ vertices in $S$. Since at most $f$ vertices fail, one of the neighbors of $x$ in $S$ ,say, $s'$ survives. Let $\pi_{H\setminus F}(u, s')$ be the $u-s'$ shortest path in $H \setminus F$, and consider the following $u-t$ path $P'=\pi_{H\setminus F}(u, s')\cdot (s',x) \cdot P_{u,t,F}[x \leadsto t]$. By the definition of $x$, $P' \subseteq H$. In addition, since $H$ contains an $f$-FT $S$-sourcewise preserver and $s'\in S$, it holds that 
\begin{eqnarray*}
\dist_{H \setminus F}(u, t)&\leq& |P'|=\dist_{G \setminus F}(u,s')+1+|P_{u,t,F}[x \leadsto t]|
\\&\leq&
\dist_{G \setminus F}(u,x)+2+P_{u,t,f}[x \leadsto t]=|P_{u,t,F}|+2=\dist_{G \setminus F}(u, t) +2.
\end{eqnarray*}
The lemma follows.
\end{proof}

%Details are given in the Appendix \ref{append:missproof}. 

%\input{FTSTPreservers}

%\input{FTSPLowerbound_new}

\section{Lower Bounds for FT Preservers and Additive Spanners}
\label{sec:lbST}
In this section, we provide the first non-trivial lower bounds for preservers and additive spanners for a single pair $s$-$t$. 
%This leads to a super-linear lower bound for FT additive spanners. 
We start by proving the following theorem. 
%show:\fabr{Put back a more precise claim (before missing even in appendix): more useful for future work}
  \begin{theorem}
\label{thr:LBadditive}
For any two integers $q,h>0$ and a sufficiently large $n$, there exists an unweighted undirected $n$-node graph $G = (V, E)$ and a pair $s,t\in V$ such that any $2hq$-FT $(2q-1)$-additive spanner for $G$ for the single pair $(s, t)$ has size $\Omega((\frac{n}{hq})^{2-2/(h+1)})$.
%No $f$-FT approximate distance sensitivity oracle, even for a single $(s, t)$ demand pair, can simultaneously have strongly subquadratic space and $+o(f)$ error. 
\end{theorem}

The main building block in our lower bound is the construction of an (undirected unweighted) tree $\tree^h$, where $h$ is a positive integer parameter related to the desired number of faults $f$. Tree $\tree^h$ is taken from \cite{ParterPODC15} with mild technical adaptations. Let $d$ be a \emph{size} parameter which is used to obtain the desired number $n$ of nodes. 
It is convenient to interpret this tree as rooted at a specific node (though edges in this construction are undirected). We next let $rt(\tree^h)$ and $L(\tree^h)$ be the root and leaf set of $\tree^h$, respectively. We also let $\ell(h)$ and $n(h)$ be the height and number of nodes of $\tree^h$, respectively.

Tree $\tree^h$ is constructed recursively as follows (see also Fig. \ref{fig:Th}). The base case is given by $\tree^0$ which consists of a single isolated root node $rt(\tree^0)$. Note that $\ell(0)=0$ and $n(0)=1$. In order to construct $\tree^h$, we first create $d$ copies $\tree^{h-1}_0,\ldots, \tree^{h-1}_{d-1}$ of $\tree^{h-1}$. Then we add a path $v_0,\ldots,v_{d-1}$ of length $d-1$ (consisting of new nodes), and choose $rt(\tree^h)=v_0$. Finally, we connect $v_j$ to $rt(\tree^{h-1}_j)$ with a path (whose internal nodes are new) of length $(d-j)\cdot(\ell(h-1)+3)$. Next lemma illustrates the crucial properties of $\tree^h$.
\begin{lemma}\label{lem:propertiesT}
The tree $\tree^h$ satisfies the following properties:
\begin{enumerate}\itemsep0pt
\item\label{prop:numberNodes} $n(h)\leq \frac{3}{2}(h+1)(d+1)^{h+1}$
\item\label{prop:numberLeaves} $|L(\tree^h)|=d^h$
\item\label{prop:faultSet} For every $\ell \in L(\tree^h)$, there exists $F_\ell\subseteq E(\tree)$,  $|F_\ell|=h$, such that 
$\dist_{\tree^h\setminus F_\ell}(s, \ell) \leq \dist_{\tree^h\setminus F_\ell}(s, \ell')+2$ for every $\ell' \in L(\tree^h)\setminus \{\ell'\}$.
\end{enumerate}
\end{lemma}

We next construct a graph $S^h$ as follows. We create two copies $\tree_s$ and $\tree_t$ of $\tree^h$. We add to $S^h$ the complete bipartite graph with sides $L(\tree_s)$ and $L(\tree_t)$, which we will call the \emph{bipartite core} $B$ of $S^h$. Observe that $|L(\tree_s)|=|L(\tree_t)|=d^h$, and hence $B$ contains $d^{2h}$ edges.
We will call $s=sr(S^h)=rt(\tree_s)$ the source of $S^h$, and $t=tg(S^h)=rt(\tree_t)$ its target. See Fig. \ref{fig:Sh} for an illustration.
%Note that \fab{crucially the bipartite core} has $\Omega(n^{2-2/(h+1)})$ edges. 

\begin{lemma}\label{lem:Sh}
Every $2h$-FT $(s,t)$ preserver (and 1-additive \fab{$(s,t)$} spanner) $H$ for $S^{h}$ must contain each edge $e=(\ell_s,\ell_t)\in B$.
\end{lemma}
\begin{proof}
Assume that $e=(\ell_s,\ell_t) \notin H$ and consider the case where $F_{\ell_s}$ fails in $\tree_s$ and $F_{\ell_t}$ fails in $\tree_t$. Let $G':=\fab{S^h} \setminus (F_{\ell_s} \cup F_{\ell_t})$, and $d_s$ (resp., $d_t$) be the distance from $s$ to $\ell_s$ (resp., from $\ell_t$ to $t$) in $G'$. By Lemma \ref{lem:propertiesT}\fab{.\ref{prop:faultSet}} the shortest $s$-$t$ path in $G'$ passes through $e$ and has length $d_s+1+d_t$. By the same lemma, any path in $G'$, hence in $H':=H \setminus (F_{\ell_s} \cup F_{\ell_t})$, that does not pass through $\ell_s$ (resp., $\ell_t$) must have length at least $(d_s+2)+1+d_t$ (resp., $d_s+1+(d_t+2)$). On the other hand, any path in $H'$ that passes through $\ell_s$ and $\ell_t$ must use at least $3$ edges of $B$, hence having length at least $d_s+3+d_t$.
 \end{proof}

Our lower bound graph $S^h_q$ (see also Fig. \ref{fig:Shq}) is obtained by taking $q$ copies $S_1,\ldots,S_q$ of graph $S^h$ with $d=(\frac{n}{3q(h+1)}\fab{-1})^{\frac{1}{h+1}}$, and chaining them with edges $(tg(S_i),sr(S_{i+1}))$, for $i=1,\ldots,q-1$. We let $s=sr(S_1)$ and $t=tg(S_q)$.
%This concludes the construction of our lower bound graph $S(h,q,n)$.
\begin{proof}[Proof of Theorem \ref{thr:LBadditive}]
Consider $S^h_q$. By Lemma \ref{lem:propertiesT}\fab{.\ref{prop:numberNodes}-\ref{prop:numberLeaves}} this graph contains at most $n$ nodes, and the bipartite core of each $S_i$ contains $d^{2h}=\Omega((\frac{n}{qh})^{2-2/(h+1)})$ edges. 

Finally, we show that any $(2q-1)$-additive $(s,t)$ spanner needs to contain all the edges of at least one such bipartite core. Let us assume this does not happen, and let $e_i$ be a missing edge in the bipartite core of $S_i$ for each $i$. Observe that each $s$-$t$ shortest path has to cross $sr(S_i)$ and $tg(S_i)$ for all $i$. Therefore, it is sufficient to choose $2h$ faulty edges corresponding to each $e_i$ as in Lemma \ref{lem:Sh}. This introduces an additive stretch of $2$ in the distance between $s$ and $t$ for each $e_i$, leading to a total additive stretch \fab{of} at least $2q$. 
\end{proof}

The same construction can also be extended to the setting of $(2h)$-FT $S\times T$ preservers. To do that, we make \emph{parallel} copies of the $S^{h}$ graph. Details are given in Appendix \ref{sec:LB_ST}.

\paragraph{Improving over the Bipartite Core} The proof above only gives the trivial lower bound of $\Omega(n)$ for the case of two faults (using $h=q=1$).
We can strengthen the proof in this special case to show instead that $\Omega(n^{1 + \eps})$ edges are needed, and indeed this even holds in the presence of a \emph{polynomial additive stretch}:
\begin{theorem} \label{thm:subs-lb-improvement}
A $2$-FT distance preserver of a single $(s, t)$ pair in an undirected unweighted graph needs $\Omega(n^{11/10 - o(1)})$ edges.
\end{theorem}
%Indeed, we are able to show that a super-linear size is needed even in the
%presence of a \emph{polynomial additive stretch}:
\begin{theorem} \label{thm:gap-pair-lb}
There are absolute constants $\eps, \delta > 0$ such that any $+n^{\delta}$-additive $2$-FT preserver for a single $(s, t)$ pair in an undirected unweighted graph needs $\Omega(n^{1 + \eps})$ edges.
\end{theorem}

Finally, by tolerating one additional fault, we can obtain a strong incompressibility result:
\begin{theorem} \label{thm:gap-compression-lb}
There are absolute constants $\eps, \delta > 0$ such that any $+n^{\delta}$-additive $3$-FT distance sensitivity oracle for a single $(s, t)$ pair in an undirected unweighted graph uses $\Omega(n^{1 + \eps})$ bits of space.
\end{theorem}

The proofs of Theorems \ref{thm:subs-lb-improvement}, \ref{thm:gap-pair-lb} and \ref{thm:gap-compression-lb} are all given in Appendix \ref{sec:apxLB}.
The central technique in their proofs, however, is the same.
The key observation is that the structure of $\tree_s, \tree_t$ allows us to use our faults to select leaves $\ell_{\fab{s}}, \ell_t$ and enforce that a shortest $\ell_{\fab{s}}-\ell_t$ path is kept in the graph.
When we use a bipartite core between the leaves of $\tree_s$ and $\tree_t$, this ``shortest path'' is simply an edge, so the quality of our lower bound is equal to the product of the leaves in $\tree_s$ and $\tree_t$. However, sometimes a better graph can be used instead.
In the case $h=1$, we can use a nontrivial lower bound graph against (non-faulty) subset distance preservers (from \cite{B17}), which improves the cost per leaf pair from $1$ edge to roughly $n^{11/10}$ edges, yielding Theorem \ref{thm:subs-lb-improvement}.
Alternatively, we can use a nontrivial lower bound graph against $+n^{\delta}$ spanners (from \cite{AmirGreg}), which implies Theorem \ref{thm:gap-pair-lb}.
The proof of \fab{Theorem} \ref{thm:gap-compression-lb} is similar in spirit, but requires an additional trick in which \emph{unbalanced trees} are used: we take $\tree_s$ as a copy of $\tree^1$ and $\tree_t$ as a copy of $\tree^2$, and this improved number of leaf-pairs is enough to push the incompressibility argument through.

\section{FT Pairwise Preservers for Weighted Graphs}
\label{sec:Weighted}
We now turn to consider weighted graphs, for which the space requirements for FT $(s,t)$ preservers are considerably larger. %Omitted proofs and related results are given in Appendix \ref{sec:apxWeighted}.
%In Appendix \ref{sec:apxWeighted}, we show:
\begin{theorem} \label{thm:uwstub}
For any undirected weighted graph $G$ and pair of nodes $(s, t)$, there is a $1$-FT $(s,t)$ preserver with $O(n)$ edges.
\end{theorem}
To prove Thm. \ref{thm:uwstub}, we first need:
%\begin{theorem} %\label{thm:uwstub}
%For any undirected weighted graph $G$ and pair of nodes $(s, t)$, there is a $1$-FT $s$-$t$ preserver with $O(n)$ edges.
%\end{theorem}
\begin{lemma} \label{lem:weightedpathdecomp}
In an undirected weighted graph $G$, for any replacement path $P_{s, t, e}$ protecting against a single edge fault, there is an edge $(x, y) \in P_{s, t, e}$ such that there is no shortest path from $s$ to $x$ in $G$ that includes $e$, and there is no shortest path from $t$ to $y$ in $G$ that includes $e$.
\end{lemma}
\begin{proof}
Let $x$ be the furthest node from $s$ in $P_{s, t, e}$ such that there is no shortest path from $s$ to $x$ in $G$ that includes $e$.  Note that if $x = t$ then there is no path from $s$ to $t$ that uses $e$ and so the claim holds trivially.  We can therefore assume $x \ne t$, and define: let $y$ be the node immediately following $x$ in $P_{s, t, e}$.  It must then be the case that there is a shortest path from $s$ to $y$ that includes $e$.

Let $e =: (u, v)$, with $\dist(s, u) < \dist(s, v)$.  The shortest path from $s$ to $y$ that uses $e$ must then intersect $u$ before $v$, so we have
$\dist(u, y) > \dist(v, y).$
Thus, any shortest path in $G$ beginning at $y$ that uses $(u, v)$ must intersect $v$ before $u$.  However, we have $\dist(u, t) > \dist(v, t)$.  Therefore, any shortest path ending at $t$ that uses $(u, v)$ must intersect $u$ before $v$.  It follows that any shortest path beginning at $y$ and ending at $t$ does not use $(u, v)$.
\end{proof}
We can now prove:
\begin{proof} [Proof of Theorem \ref{thm:uwstub}]
To construct the preserver, simply add shortest path trees rooted at $s$ and $t$ to the preserver.
If the edge fault $e$ does not lie on the included shortest path from $s$ to $t$, then the structure is trivially a preserver.  Thus, we may assume that $e$ is in $\pi(s, t)$.  We now claim that, for some valid replacement path $P_{s, t, e}$ protecting against the fault $e$, all but one (or all) of the edges of $P_{s, t, e}$ are in the preserver.  To see this, we invoke Lemma \ref{lem:weightedpathdecomp}: there is an edge $(x, y)$ in $P_{s, t, e}$ such that no shortest path from $s$ to $x$ and no shortest path from $t$ to $y$ in $G$ $P_{s, t, e}[s \leadsto x]$ uses $e$.  Therefore, our shortest path trees rooted at $s$ and $t$ include a shortest path from $s$ to $x$ and from $t$ to $y$, and these paths were unaffected by the failure of $e$.  Therefore, $P_{s, t, e}$ has all edges in the preserver, except possibly for $(x, y)$. There are at most $n$ edges on $\pi(s, t)$, so there are at most $n$ edge faults for which we need to include a replacement path in our preserver.  We can thus complete the preserver by adding the single missing edge for each replacement path, and this costs at most $n$ edges.
If the edge fault $e$ does not lie on the included shortest path from $s$ to $t$, then the structure is trivially a preserver.  Thus, we may assume that $e$ is in $\pi(s, t)$.  We now claim that, for some valid replacement path $P_{s, t, e}$ protecting against the fault $e$, all but one (or all) of the edges of $P_{s, t, e}$ are in the preserver.  To see this, we invoke Lemma \ref{lem:weightedpathdecomp}: there is an edge $(x, y)$ in $P_{s, t, e}$ such that no shortest path from $s$ to $x$ and no shortest path from $t$ to $y$ in $G$ $P_{s, t, e}[s \leadsto x]$ uses $e$.  Therefore, our shortest path trees rooted at $s$ and $t$ include a shortest path from $s$ to $x$ and from $t$ to $y$, and these paths were unaffected by the failure of $e$.  Therefore, $P_{s, t, e}$ has all edges in the preserver, except possibly for $(x, y)$. There are at most $n$ edges on $\pi(s, t)$, so there are at most $n$ edge faults for which we need to include a replacement path in our preserver.  We can thus complete the preserver by adding the single missing edge for each replacement path,, paying $\leq n$ edges.
\end{proof}
%With a simple union bound, we get that any set $P$ of node pairs can be preserved using $O(\min(n|P|, n^2))$ edges.  It is natural to wonder if a better bound is possible when the pair set is larger.  Surprisingly, the answer is no: we provide a matching lower bound for Theorem \ref{thm:uwstub}.
%\begin{theorem}\label{thr:UBqpairs}
%For any $q \le n$, there exists an undirected weighted graph $G$ and a set $Q$ of $q$ node pairs such that every $1$-FT preserver of $G, Q$ has $\Omega(n|Q|)$ edges.
%\end{theorem}
With a trivial union bound, we get that any set $P$ of node pairs can be preserved using $O(\min(n|P|, n^2))$ edges.  It is natural to wonder if one can improve this union bound by doing something slightly smarter in the construction.

%Surprisingly, the answer is \fab{NO}: we are able to provide a matching lower bound for Theorem \ref{thm:uwstub}.\fabr{Moved to appendix proof of Thr \ref{thr:UBqpairs} to make room to the next thr. claims that were appearing only in appendix (not good)}
%\begin{theorem}\label{thr:UBqpairs}
%For any $q \le n$, there exists an undirected weighted graph $G$ and a set $Q$ of $q$ node pairs such that every $1$-FT preserver of $G, Q$ has $\Omega(n|Q|)$ edges.
%\end{theorem}
\begin{theorem} 
\label{thr:UBqpairs}
For any integer $1\leq p \leq {n\choose 2}$, there exists an undirected weighted graph $G$ and a set $P$ of $p$ node pairs such that every $1$-FT $P$-pairwise preserver of $G$ contains $\Omega(\min(np, n^2))$ edges.
\end{theorem}
\begin{proof}
%This proof is an adaptation of Lemma \ref{lem:sizeG}.%\cite{PPFTBFS13}.
We construct our lower bound instance by adapting the construction in Lemma \ref{lem:Sh}.  First, add a path of length $n+1$ using edges of weight $1$.  Call the nodes on the path $p_1, \dots, p_{n+1}$.  Next, create $n$ new nodes $\{v_i\}$, and add an edge of weight $1$ from $p_{n+1}$ to each $v_i$.  Then, for each $i \in [1, p]$, add a new node $x_i$ to the graph, and connect $x_i$ to $p_i$ with an edge of weight $2(n - i) + 1$.  Finally, for all $i \in [1, p], j \in [1, n]$, add an edge of weight $1$ between $x_i$ and $v_j$. Define the pair set $P$ to be $\{s\} \times \{v_i \ \mid \ i \in [1, p]\}$.  Note that the graph has $\Theta(n)$ nodes and $\Omega(n|P|)$ edges, because there are exactly $n|P|$ edges between the nodes $\{x_i\}$ and $\{v_j\}$.  We will complete the proof by arguing that all edges in $\{x_i\} \times \{v_j\}$ must be kept in the preserver.  Specifically, we claim that for any $i, j$, the edge $(x_i, v_j)$ is needed to preserve the distance of the pair $(s, v_j) \in P$ when the edge $(p_i, p_{i+1})$ faults.  To see this, note that any path from $s$ to $v_j$ must pass through {\em some} node $x_i$, and we have $\dist(s, x_i) = (i-1) + 2(n-i) + 1 = 2n - i$ for any $i$.  Since $(p_i, p_{i+1})$ has faulted, the path from $s$ to $v_j$ must intersect $x_{i'}$ for some $i' \le i$ before it intersects $x_{i''}$ for any $i'' > i$.  Therefore, the shortest $s-v_j$ path passes through $x_i$, and thus uses $(x_i, v_j)$.
\end{proof}

%Many remaining results for weighted FT preservers are in Appendix \ref{sec:apxWeighted}.
 
We show that the situation dramatically changes for $f=2$.
\begin{theorem}\label{thr:LBweighted2FT}
There exists an undirected weighted graph $G$ and a single node pair $(s, t)$ in this graph such that every $2$-FT $(s,t)$ preserver \fab{of $G$} requires $\Omega(n^2)$ edges. \fab{The same lower bound holds on the number of bits of space used by any exact distance sensitivity oracle in the same setting.}
\end{theorem}
\begin{proof}
For the first claim, we construct our lower bound instance as follows.  Build node disjoint paths $P_s\fab{=(s=s_0,s_1\ldots,s_{n-1})}$ and $P_t\fab{=(t=t_0,t_1,\ldots,t_{n-1})}$ of $n$ \fab{nodes} each.
All of the edges in these paths have weight zero (or sufficiently small $\eps > 0$ will do). 
Next, we add a complete bipartite graph $X \times Y$ \fab{with edges of weight $1$},\fabr{Weight 1 avoids going around in the bipartite core} where $X = \{x_0,\ldots,x_{n-1}\}$ and $Y = \{y_0,\ldots,y_{n-1}\}$ are new node sets of size $n$ each.
Finally, for each $i \in \{0, \ldots, n-1\}$, we add edges $(s_i,x_i)$ and $(t_i,y_i)$ of weight \fab{$n-i$}. See Fig. \ref{fig:lowerboundtwof:left} for an illustration of this construction.

%\begin{wrapfigure}{r}{0.29\textwidth}
%\vspace{-20pt}
  %\begin{center}
    %\includegraphics[width=0.32\textwidth]{weighted2ft.pdf}
  %\end{center}
	%\vspace{-10pt}
%\end{wrapfigure}

We now claim that every $2$-FT $s$-$t$ preserver must include all edges of the bipartite graph $X \times Y$. In more detail, the edge $(x_i,y_j)$ is needed when the edges $e^{\fab{s}}_i=(s_i, s_{i+1})$ and $e^{\fab{t}}_j=(t_j, t_{j+1})$ fail. Indeed, there is path of length $n-i+n-j+1$ passing throw $(x_i,y_j)$ in $G \setminus \{e^s_i,e^t_j\}$ and any other $s$-$t$ path has length at least $n-i+n-j+2$. The first claim follows.  

For the second claim, consider the same graph as before, but with possibly some missing edges in $X\times Y$. Consider any distance sensitivity oracle for this family of instances. By querying the $s$-$t$ distance for faults $(e^{\fab{s}}_i,e^{\fab{t}}_j)$, one obtains $n-i+n-j+1$ iff the edge $(x_i,y_j)$ is present in the input graph. This way it is possible to reconstruct the edges $E'\subseteq X\times Y$ in the input instance. Since there are $\Omega(2^{n^2})$ possible input instances, the size of the oracle has to be $\Omega(n^2)$.
\end{proof}
We next consider the case of directed graphs, and prove Theorem \ref{thr:UBLB1FT}. We split its proof in the next two lemmas.  Let $\DP(n)$ describe the worst-case sparsity of a (non-FT) preserver of $n$ node pairs in a directed weighted graph.  That is, for any directed weighted $n$-node graph $G$ and set $P$ of $|P| = n$ node pairs, there exists a distance preserver of $G, P$ on at most $\DP(n)$ edges, yet there exists a particular $G, P$ for which every distance preserver has $\DP(n)$ edges.

\begin{lemma} \label{lem:dwstub} 
Given any $s$-$t$ pair in a directed weighted graph, there is a $1$-FT $s$-$t$ preserver whose sparsity is $O(\DP(n))$.
\end{lemma}
\begin{proof}
Add a shortest path $\pi(s, t)$ to the preserver, and note that we only need replacement paths in our preserver for edge faults $e$ on the path $\pi(s, t)$.  There are at most $n-1$ such edges; thus, the preserver is the union of at most $n-1$ replacement paths.  For each replacement path $P_{s, t, e}$, note that the path is disjoint from $\pi(s, t)$ only on one continuous subpath.  Let $a, b$ be the endpoints of this subpath.  Then $P_{s, t, e}[a \leadsto b]$ is a shortest path in the graph $G \setminus \pi(s, t)$, and all other edges in $P_{s, t, e}$ belong to $\pi(s, t)$.  Therefore, if we include in the preserver all edges in a shortest path from $a$ to $b$ in $G \setminus \pi(s, t)$, then we have included a valid replacement path protecting against the edge fault $e$.
By applying this logic to each of the $n-1$ possible edge faults on $\pi(s, t)$, we can protect against all possible edge faults by building any preserver of $n-1$ node pairs in the graph $G \setminus \pi(s, t)$.
\end{proof}

\begin{lemma} \label{lem:dwstlb}
There is a directed weighted graph $G$ and a node pair $s$-$t$ such that any $1$-FT $s$-$t$ preserver requires $\Omega(\DP(n))$ edges.
\end{lemma}
\begin{proof}
Let $K$ be a directed graph on $O(n)$ nodes and nonnegative edge weights, and let
$P = \{(x_1,y_1),...,(x_n,y_n)\}$ be a set of node pairs of size $n$ such that the sparsest preserver of $K, P$ has $\Omega(f(n))$ edges.

Add to $K$ a directed path $Y := (s \to a_1 \to b_1 \to a_2 \to b_2 \to \dots \to a_n \to b_n \to t)$ on $2n+2$ new nodes.
All edges in $Y$ have weight $0$ (or sufficiently small $\eps > 0$ will do).  All other edge weights in the graph will be nonnegative, so $Y$ is the unique shortest path from $s$ to $t$.

Let $M$ be the largest edge weight in $K$, and let $W=M \cdot n$. Note that $W$ is larger than the weight of any shortest path in $K$.  Now, for each $i \in [1, n]$, add an edge from $a_i$ to $x_i$ of weight $(n-i)W$ and add an edge from $y_i$ to $b_i$ of weight $iW$. See Fig. \ref{fig:lowerboundtwof:right}  for an illustration
This completes the construction of $G$.  There are $O(n)$ nodes in $G$, and so it suffices to show that all edges in a preserver of $K, P$ must remain in a $1$-FT $s$-$t$ preserver for $G$.

Let $(a_i \to b_i)$ be an edge on the path $Y$.
If this edge faults, then the new shortest path from $s$ to $t$ has the following general structure: the path travels from $s$ to $a_j$ for some $j \le i$, then it travels to $x_j$, then it travels a shortest path from $x_j$ to $y_k$ (for some $k \ge i$) in $K$, then it travels from $y_k$ to $b_k$, and finally it travels from $b_k$ to $t$.
The length of this path is then
$$(n-j)W + \dist_K(x_j,y_k) + kW = nW +(k-j)W + \dist_K(x_j,y_k).$$
Suppose that $k-j \geq 1$. Then the weight of the detour is at least $(n+1)W$ (as all distances in K are nonnegative).
On the other hand, if $k=j$ (and hence $=i$), the weight of the detour is
$$nW +\dist_K(x_i,y_k) <(n+1)W$$
because we have $W > \dist_K(x_i, y_k)$.
Thus, any valid replacement path $P_{s, t, (a_i, b_i)}$ travels from $s$ to $a_i$, then from $a_i$ to $x_i$, then along some shortest path in $K$ from $x_i$ to $y_i$, then from $y_i$ to $b_i$, and finally from $b_i$ to $t$.

Hence, any $1$-FT $s$-$t$ preserver includes a shortest path from $x_i$ to $y_i$ in $K$ for all $i \in [1, n]$.  Therefore, the number of edges in this preserver is at least the optimal number of edges in a preserver of $K, P$; i.e. $\DP(n)$ edges.  Since $G$ has $O(n)$ nodes, the theorem follows.
\end{proof}

\begin{figure}
\centering
\begin{subfigure}{.4\textwidth}
  \centering
  \includegraphics[width=.8\linewidth]{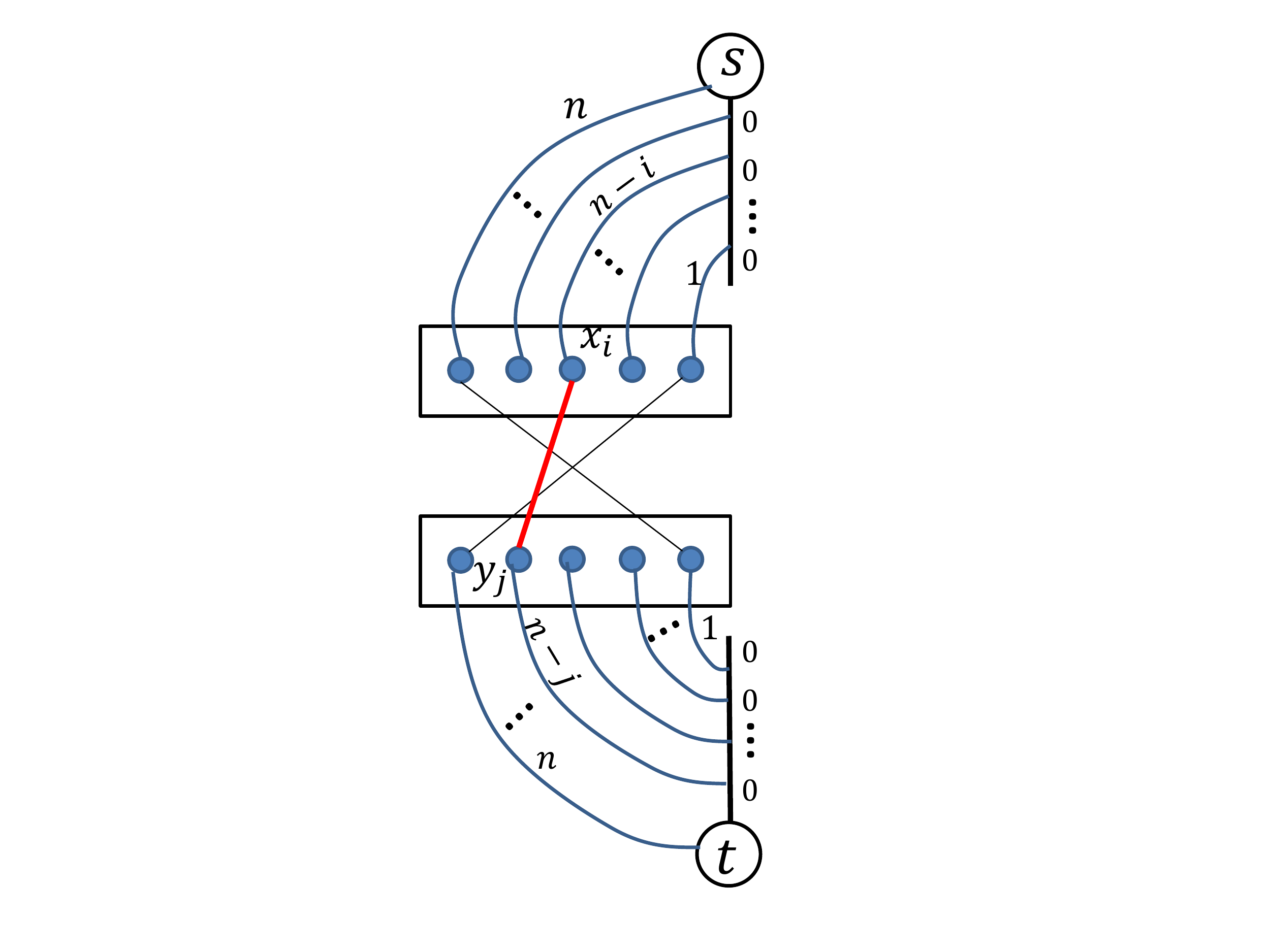}
  \caption{Lower bound example for single pair and two faults. When the $i$th edge on the $s$-path fails, the new shortest path to $t$ must go through $x_i$. Similarly, when the $j$th edge on the $t$-path fails, the shortest path to $s$ must go through $y_j$. Hence, the shortest path in $G \setminus \{e_i,e_j\}$ uses the edge $(x_i,y_j)$.}
  \label{fig:lowerboundtwof:left}
\end{subfigure}%
\hspace{1cm}
\begin{subfigure}{.4\textwidth}
  \centering
  \includegraphics[width=.8\linewidth]{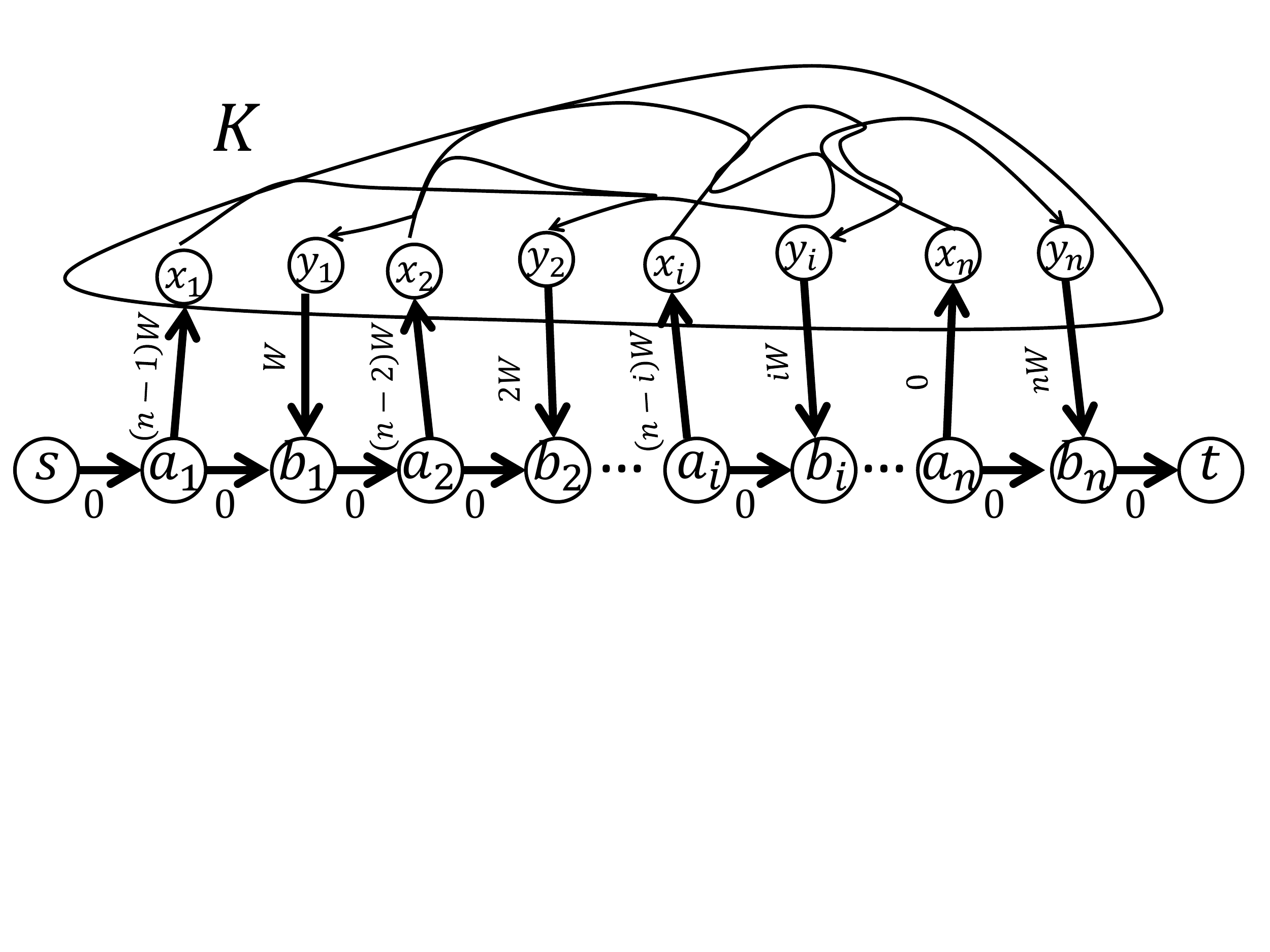}
  \caption{Lower bound construction for a single pair in the directed weighted case.  Here, $K$ is an arbitrary lower bound graph for (non-FT) distance preservers of $n$ pairs in the directed weighted setting, which can be modularly substituted in to our construction.}
  \label{fig:lowerboundtwof:right}
\end{subfigure}
\caption{Lower bounds for weighted graphs.}
\label{fig:lowerboundtwof}
\end{figure}

\section{Open Problems}
%The current work provides the first subquadratic bounds for $f$-FT preservers and additive spanners with $f \geq 2$ vertex faults. 
%It also provides the first lower bounds for the additive setting and for single pair preservers in weighted directed and undirected graphs. 

There are lots of open ends to be closed. Perhaps the main open problem is to resolve the current gap for $f$-FT single-source preservers. Since the lower bound of $\Omega(n^{2-1/(f+1)})$ edges given in \cite{ParterPODC15} has been shown to be tight for $f \in [1,2]$, it is reasonable to believe that this is the right bound for $f \geq 3$. Another interesting open question involves lower bounds for FT additive spanners. Our lower-bounds are super linear only for $f\geq 2$. The following basic question is still open though: is there a lower bound of $\Omega(n^{3/2+\epsilon})$ edges for some $\epsilon \in (0,1]$ for $2$-additive spanners with \emph{one} fault? Whereas our lower bound machinery can be 
adapted to provide non trivial bounds for different types of $f$-FT $P$-preservers (e.g., $P=\{s,t\}, P=S \times T$, etc.), our upper bounds technique for general $f\geq 2$ is still limited to the sourcewise setting. Specifically, it is not clear how to construct an $f$-FT $S \times S$ preservers other than taking the (perhaps wasteful) $f$-FT $S$-sourcewise preservers.
As suggested by our lower bounds, these questions are interesting already for a single pair.
\newpage

\bibliographystyle{alpha}
\bibliography{ftref}

\appendix

%\newpage
\section{Tables}
\label{sec:tabFig}

\begin{tabular}{ |c|c|c|c| }
%\begin{center}
\hline
\multicolumn{4}{ |c| }{$f$-VFT (or EFT) Preservers} \\
\hline
& \shortstack{$S$-sourcewise \\(unweighted)} & \shortstack{pairs $P$, $f=1$ \\ (weighted)} & \shortstack{Single pair $s$-$t$, $f=1$\\(weighted+directed)} \\ \hline
\multirow{2}{*}{\shortstack{Upper \\ Bound}} & $\widetilde{O}(f|S|^{1/2^f} \cdot n^{2-1/2^f})$ & $O(\min\{n^2,n|P|\})$ & $O(DP(n))=O(n^{3/2})$\\
 & \mbox{\textbf{New}} & \mbox{\textbf{New}} 
 & \mbox{\textbf{New}} \\ \hline
\multirow{2}{*}{\shortstack{Lower \\ Bound}} & $\Omega(|S|^{1/(f+1)} \cdot n^{2-1/(f+1)})$ & $\Omega(\min\{n^2,n|P|\})$ & $\Omega(DP(n))=\Omega(n^{4/3})$ \\
 & \cite{ParterPODC15} & \mbox{\textbf{New}} & \mbox{\textbf{New}}\\
\hline
%\end{center}
\end{tabular}

\section{Omitted Details of Section \ref{sec:lbST}}
\label{sec:apxLB}

\begin{figure}
\centering
\begin{subfigure}{.5\textwidth}
  \centering
  \includegraphics[width=.8\linewidth]{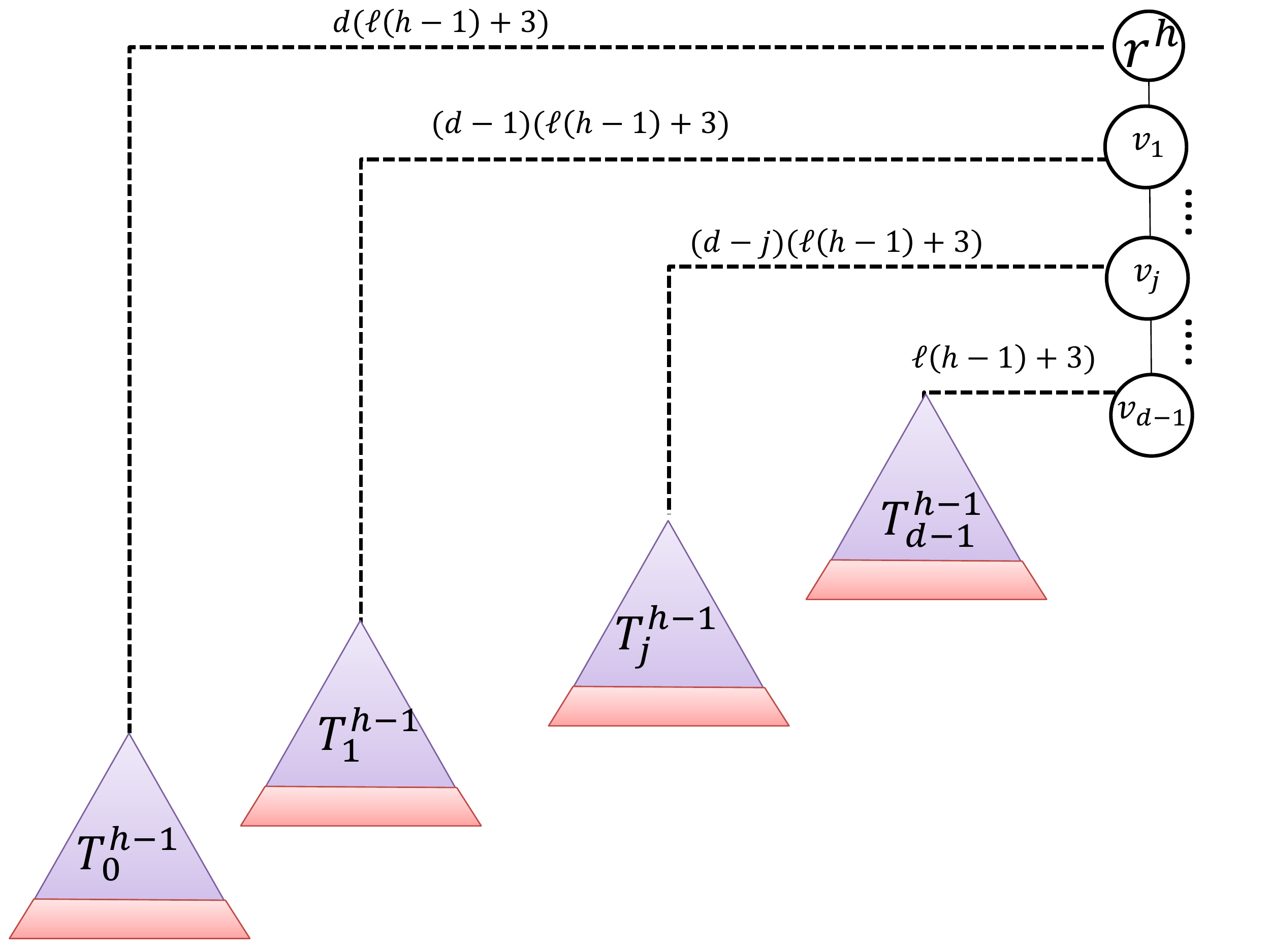}
  \caption{Structure of $\tree^h$ tree}
  \label{fig:Th}
\end{subfigure}%
\begin{subfigure}{.5\textwidth}
  \centering
  \includegraphics[width=.8\linewidth]{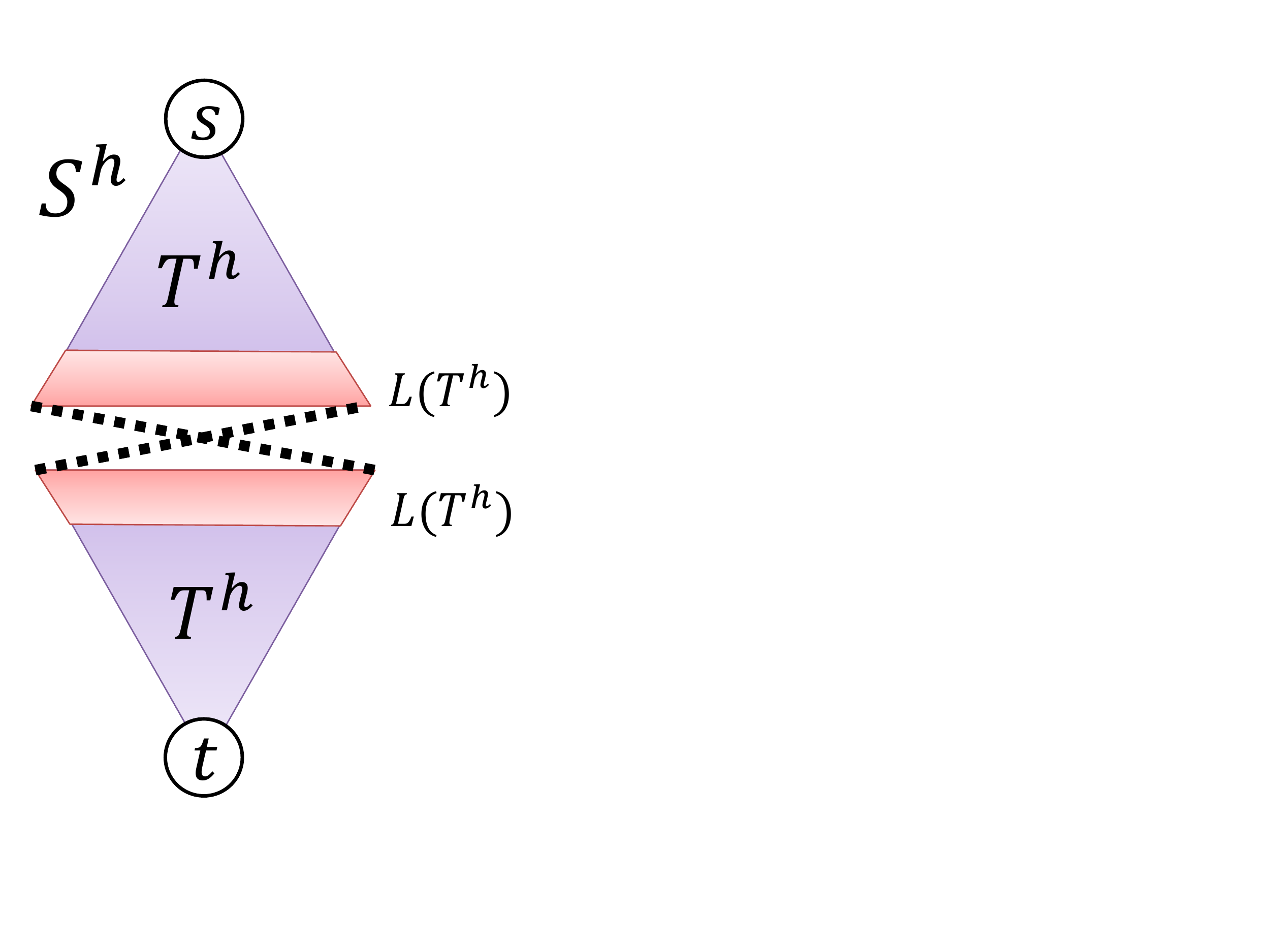}
  \caption{The graph $S^h$ used to provide a lower bound  for $s$-$t$ pair and $1$-additive stretch}
  \label{fig:Sh}
\end{subfigure}
\begin{subfigure}{.4\textwidth}
  \centering
  \includegraphics[width=.8\linewidth]{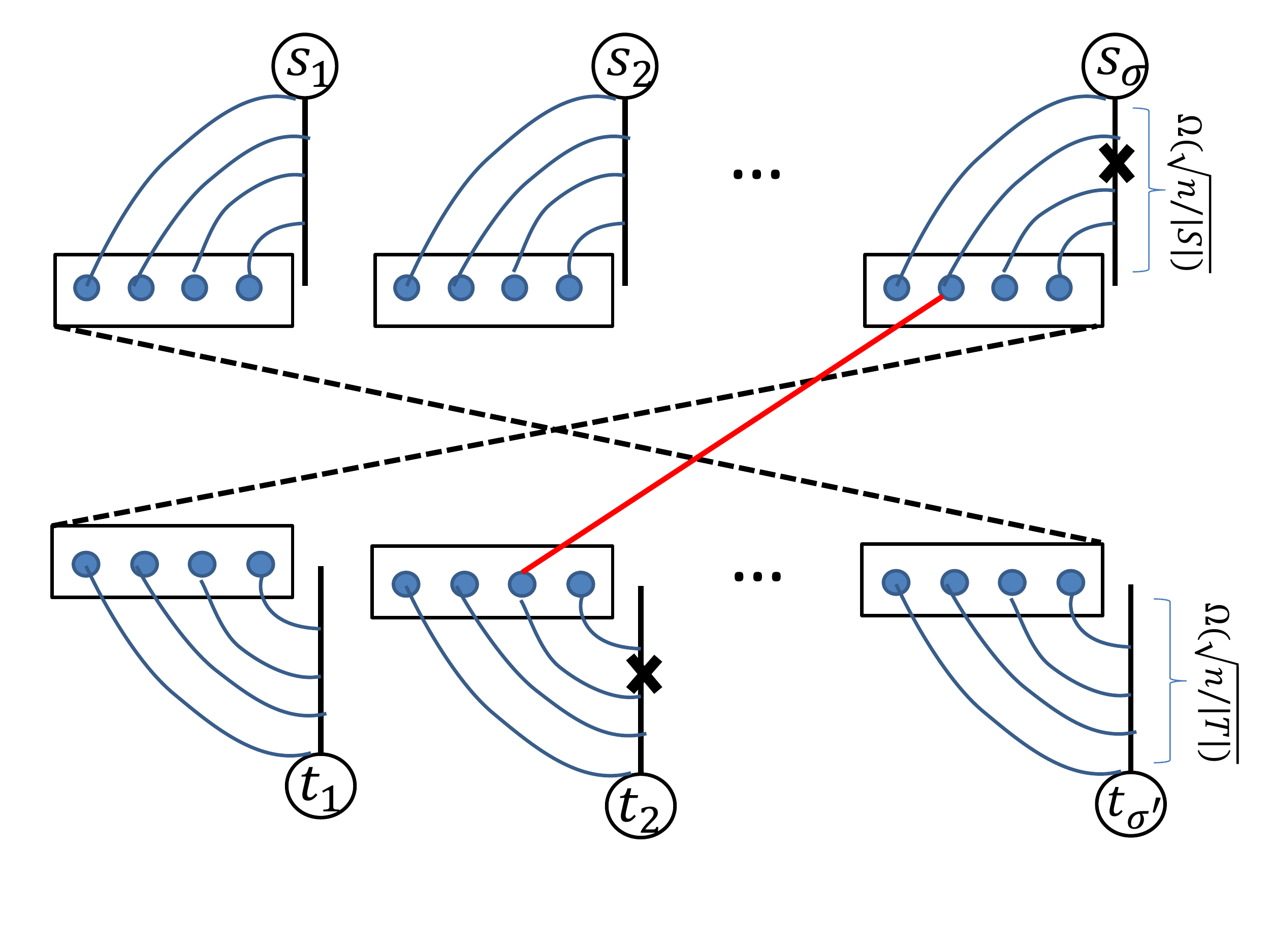}
  \caption{Lower bound for $S \times T$ preserver (and $1$-additive) for the case $f=2$}
  \label{fig:ST}
\end{subfigure}
\begin{subfigure}{.4\textwidth}
  \centering
  \includegraphics[width=.8\linewidth]{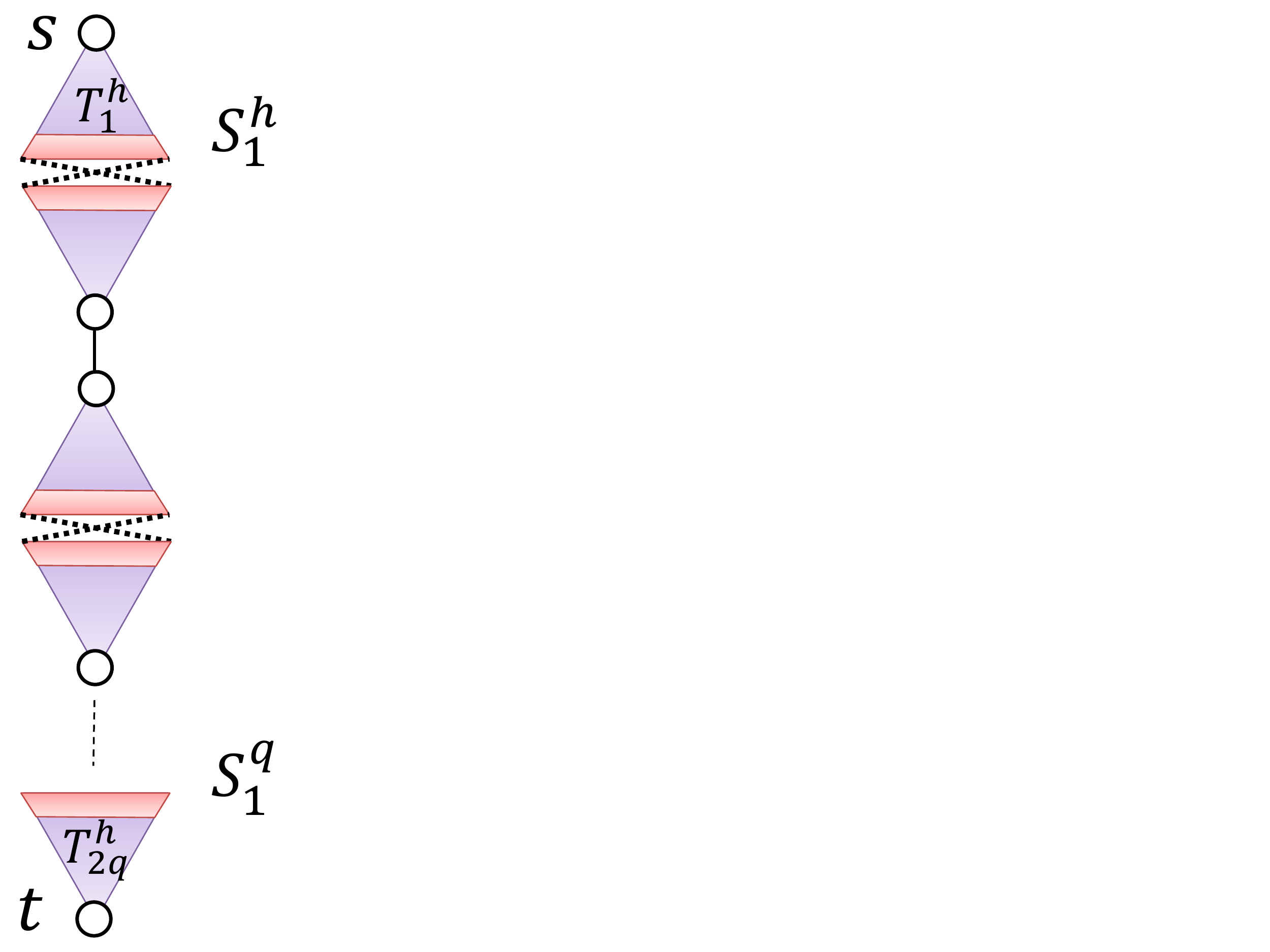}
  \caption{The graph $S^h_q$ used to provide lower bound for additive stretch $\beta=2q-1$ already for a single pair $s$-$t$}
  \label{fig:Shq}
\end{subfigure}
\caption{Illustration of the lower bound constructions.}
\label{fig:lowerboundBasic}
\end{figure}

\subsection{Properties of $\tree^h$}

We next prove Lemma \ref{lem:propertiesT}.
%The following useful facts about $\tree^h$ are analogous to Obs. 4.2 and Lemma 4.3 from 
%\cite{parter2015dual}; 
\begin{lemma}\label{lem:heightT-nodesT}
$\ell(h)=3((d+1)^h-1)$ and $n(h)\leq \frac{3}{2}(h+1)(d+1)^{h+1}$.
\end{lemma}
\begin{proof} %\emph{[of Lemma \ref{lem:heightT-nodesT}]}
Let us prove the first claim by induction on $h$. The claim is trivially true for $h=0$. Next suppose it holds up to $h-1\geq 0$, and let us prove it for $h$. All the subtrees $\tree^{h-1}_j$ used in the construction of $\tree^h$ have the same height $\ell(h-1)$ which is $3((d+1)^{h-1}-1)$ by the inductive hypothesis. The distance between $rt(\tree^h)$ and $rt(\tree^{h-1}_j)$ is $j+(d-j)(\ell(h-1)+3)$ which is a decreasing function of $j$. In particular, the maximum such distance is the one to $rt(\tree^{h-1}_0)$, which is $d(\ell(h-1)+3)$. We conclude that
$$
\ell(h)=\ell(h-1)+d(\ell(h-1)+3)=3d+(d+1)3\left((d+1)^{h-1}-1\right)=3(d+1)^h-3.
$$

The second claim is trivially true for $h=0$. The number of nodes in $\tree^h$ is given by $d$ times the number of nodes in $\tree^{h-1}$, plus the sum of the \emph{lengths} of the paths connecting each $v_j$ to $rt(\tree^{h-1}_j)$, i.e. 
\begin{align*}
n(h) & =d\cdot n(h-1)+\sum_{j=0}^{d-1}(d-j)(\ell(h-1)+3) \overset{\substack{first\\ claim}}{=} d\cdot n(h-1)+3(d+1)^{h-1}\frac{d(d+1)}{2}\\
& \overset{\substack{Induct.\\ hypoth.}}{\leq} d\cdot \frac{3}{2}h(d+1)^{h}+\frac{3}{2}d(d+1)^{h} =\frac{3}{2}(h+1)d(d+1)^h\leq \frac{3}{2}(h+1)(d+1)^{h+1}.
\end{align*}
\end{proof}

We next need a more technical lemma which will be useful to analyze the stretch. An easy inductive proof shows that $|L(\tree^h)|=d^h$. It is convenient to sort these leaves from left to right using the following inductive process.  The base case is that $\tree^0$ has a unique leaf (the root) which obviously has a unique ordering.
For the inductive step, given the sorting of the leaves of $\tree^{h-1}$, the sorting for $\tree^h$ is achieved by placing all the leaves in the subtree $\tree^{h-1}_j$ to the left of the leaves of the subtree $\tree^{h-1}_{j+1}$, $j=0,\ldots,d-2$ (the leaves of each $\tree^{h-1}_j$ are then sorted recursively). Given this sorting, we will name the leaves of $\tree^H$ (from left to right) $\ell^h_0,\ldots,\ell^h_{d^h-1}$. 

For each leaf $\ell^h_j$, we next recursively define a subset of at most $h$ (faulty) edges $F^h_j$. Intuitively, these are edges that we can remove to make $\ell^h_j$ the closest leaf to the root. We let $F^0_j=\emptyset$. Suppose that $\ell^h_j$ is the $r$-th leaf from left to right of $\tree^{h-1}_t$ (in zero-based notation). Then $F^h_j$ is given by the edges of type $F^{h-1}_r$ in $\tree^{h-1}_t$, plus edge $(v_t,v_{t+1})$ if $t<d-1$. Note that obviously $|F^h_j|\leq h$.
\begin{lemma}\label{lem:faultyT}
One has that $\ell^h_j$ is the leaf at minimum finite distance $d':=\dist_{\tree^h - F^h_j}(rt(\tree^h),\ell^h_j)$ from $rt(\tree^h)$ in $\tree^h- F^h_j$, and any other leaf in $L(\tree^h)-\{\ell^h_j\}$ is at distance at least $d'+2$ from $rt(\tree^h)$. 
\end{lemma}
\begin{proof}
Once again the proof is by induction. Let $r^h:=rt(\tree^h)$. The claim is trivially true for $h=0$ since there is a unique leaf $\ell^0_0=r^0$. Next assume the claim is true up to $h-1\geq 0$, and consider $\tree^h$. Consider any leaf $\ell^h_j$, and with the same notation as before assume that it is the $r$-th leaf from left to right of $\tree^{h-1}_t$. Observe that by removing edge $(v_t,v_{t+1})$ we disconnect from $r^h$ all nodes in subtrees $\tree^{h-1}_{t'}$ with $t'>t$. In particular the distances from $r^h$ to the leaves of those subtrees becomes unbounded. Next consider a leaf $\ell'$ in a tree $\tree^{h-1}_{t'}$ with $t'<t$. By construction we have that 
\begin{align*}
\dist_{G - F^h_j}(r^h,\ell') &\geq \dist_G(r^h,\ell')\\
&\geq \dist_G(r^h,rt(\tree^{h-1}_{t'}))\\
&= t'+(d-t')(\ell(h-1)+3)\\
&\geq (t-1)+(d-t+1)(\ell(h-1)+3)\\
&= t+(d-t)(\ell(h-1)+3)+\ell(h-1)+2.
\end{align*}
On the other hand, any leaf in $L(\tree^{h-1}_t)$ which is still connected to $rt(\tree^{h-1}_t)$ has distance at most $t+(d-t)(\ell(h-1)+3)+\ell(h-1)$ from $r^h$. Recall that we are removing the edges of type $F^{h-1}_r$ from $\tree^{h-1}_t$. Note also that $\ell^h_j$ corresponds to leaf $\ell^{h-1}_r$ in $\tree^{h-1}_t$. Hence, by the inductive hypothesis, $\ell^h_j$ is the leaf of $\tree^{h-1}_t$ at minimum finite distance $d'$ from $r^{h-1}_t$, and any other such leaf is at distance at least $d'+2$ from $r^{h-1}_t$. The claim follows.  
\end{proof}

\begin{proof}[Proof of Lemma \ref{lem:propertiesT}]
Claim \ref{prop:numberNodes} is given by Lemma \ref{lem:heightT-nodesT}, Claim \ref{prop:numberLeaves} by a trivial induction, and Claim \ref{prop:faultSet} by Lemma \ref{lem:faultyT}. 
\end{proof}

\subsection{Improvement with Preserver Lower Bounds}

We next prove Theorem \ref{thm:subs-lb-improvement}. We need the following technical lemma.
%We show:\fabr{This should be Thr \ref{thm:subs-lb-improvement}}
%\begin{theorem} \label{thm:f2-preserver-improvement}
%A $2$-FT distance preserver for a single $(s, t)$ pair needs $\Omega(n^{11/10 - o(1)})$ edges in the worst case.
%\end{theorem}
%To prove this theorem, we need the following technical lemma:
\begin{lemma} [\cite{B17}] \label{lem:subs-pres-lb}
For all $n$, there is an undirected unweighted bipartite graph $G = (V, E)$ on $n$ nodes and $\Omega(n^{11/10 - o(1)})$ edges, as well as disjoint node subsets $S, T \subseteq V$ with $|S| = |T| = \Theta(n^{1/2})$ such that the following properties hold:
\begin{itemize}
\item For each edge $e \in E$, there is a pair of nodes $s \in S, t \in T$ with $\dist(s, t) = L$ (for some parameter $L$) such that every shortest $(s, t)$ path includes $e$.
\item For all $s \in S, t \in T$, we have $\dist(s, t) \ge L$.
\end{itemize}
\end{lemma}

The construction for Theorem \ref{thm:subs-lb-improvement} proceeds as follows.
By Lemma \ref{lem:heightT-nodesT}, the number of leaves in $\tree^1$ is $\ell = \Theta(d)$ and the number of nodes in $\tree^1$ is $n = \Theta(d^2)$, so we have $\ell = \Theta(n^{1/2})$.
As before, let $\tree_s, \tree_t$ be copies of $\tree^1$ rooted at $s, t$ respectively.
Now, let $H$ be a graph drawn from Lemma \ref{lem:subs-pres-lb}, with node subsets $S, T$, where the number of nodes $n_H$ is chosen such that $|S| = |T| = \ell(s) = \ell(t)$.
We add a copy of $H$ to the graph $\tree_s \cup \tree_t$, where $\ell(s)$ is used as the node set $S$, $\ell(t)$ is used as the node set $T$, and $O(n)$ new nodes are introduced to serve as the remaining nodes in $H$.
Note that the new graph $G = \tree_s \cup \tree_t \cup H$ now has $N = \Theta(n)$ nodes, so it (still) has $N^{11/10 - o(1)}$ edges in its internal copy of $H$.

\begin{lemma} \label{lem:req-in-g}
In $G$, we have:
\begin{itemize}
\item For each edge $e$ in the internal copy of $H$, there exist nodes $u \in \ell(s) = S, v \in \ell(t) = T$ with $\dist(u, v) = L$ such that every shortest $(u, v)$ path (in $G$) includes $e$.
\item For all $u \in \ell(s) = S, v \in \ell(t) = T$, we have $\dist(u, v) \ge L$.
\end{itemize}
\end{lemma}
\begin{proof}
First, we observe that if any shortest $(u, v)$ path $\pi(u, v)$ in $G$ contains a node $x$ \emph{not} in the internal copy of $H$, then we have $\dist(u, v) \ge L+1$.
To see this, note that by construction any $(x, v)$ path must contain a subpath $\pi(u', v') \subseteq H$ between nodes $u' \in S, v' \in T$.
By Lemma \ref{lem:subs-pres-lb} this subpath has length at least $L$.
Since the path $\pi(u, v)$ contains $x \notin \pi(u', v')$, we then have $|\pi(u, v)| \ge L+1$.

The second point in this lemma is now immediate: if $\pi(u, v)$ is contained in $H$ then we have $\dist(u, v) \ge L$ from Lemma \ref{lem:subs-pres-lb}; if $\pi(u, v)$ is not contained in $H$ then we have $\dist(u, v) \ge L+1$ from the above argument.
For the first point, note that by Lemma \ref{lem:subs-pres-lb}, there is a pair $u \in S, v \in T$ such that $\dist_H(u, v) = L$ and every shortest $(u, v)$ path in $H$ includes $e$.
Since we then have $\dist_G(u, v) \le L$, by the above it follows that $\pi(u, v)$ is contained in $H$, and the lemma follows.
\end{proof}

We can now show:
\begin{proof} [Proof of Theorem \ref{thm:subs-lb-improvement}]
Let $e$ be any edge in the internal copy of $H$ in $G$.
By Lemma \ref{lem:req-in-g}, there is a pair of nodes $u \in \ell(s) = S, v \in \ell(t) = T$ such that every $(u, v)$ shortest path in $G$ includes $e$.
Also, by Lemma \ref{lem:faultyT}, there are faults $f_1 \in \tree_s, f_2 \in \tree_t$ such that $u$ is the leaf of $\tree_s$ at minimum distance $d'_s$ from $s$ in $\tree_s \setminus \{f_1\}$, and $v$ is the leaf of $\tree_t$ at minimum distance $d'_t$ from $t$ in $\tree_t \setminus \{f_2\}$.

Thus, under fault set $\{f_1, f_2\}$, we have
$$\dist_{G \setminus \{f_1, f_2\}}(s, t) \le d'_s + L + d'_t$$
since one possible $(s, t)$ path is obtained by walking a shortest path from $s$ to $u$, then from $u$ to $v$, then from $v$ to $t$.
Moreover, any $(s, t)$ path $Q$ in $G \setminus \{f_1, f_2\}$ that does not include $u$ (or $v$) must include some other leaf $u' \ne u \in \ell(s)$, so it has length at least
$$|Q| \ge \dist_{G \setminus \{f_1, f_2\}}(s, u') + \dist_{G \setminus \{f_1, f_2\}}(u', v') + \dist_{G \setminus \{f_1, f_2\}}(v', t)$$
(for some leaf $v' \in \tree_t$, possibly equal to $v$).
By Lemmas \ref{lem:faultyT} and \ref{lem:req-in-g}, this implies
$$|Q| \ge (d'_s + 2) + L + d'_t > \dist_{G \setminus \{f_1, f_2\}}(s, t).$$
Therefore $Q$ is a non-shortest path, and so every shortest $(s, t)$ path in $G \setminus \{f_1, f_2\}$ includes nodes $u$ and $v$, and so (by Lemma \ref{lem:req-in-g}) it includes the edge $e$.
We then cannot remove the edge $e$ without destroying all shortest $(s, t)$ shortest paths in $G \setminus \{f_1, f_2\}$, so all $N^{11/10 - o(1)}$ edges in the internal copy of $H$ must be kept in any $2$-FT distance preserver of $(s, t)$.
\end{proof}

\begin{remark}
It is natural to expect that a similar improvement to the lower bound may be possible for $h > 1$; that is, we could imagine again augmenting the bipartite core with a subset preserver lower bound.
While it is conceivable that this technique may eventually be possible, it currently does not work: for $h > 1$ we have $\ell(s) = \Omega(n(s)^{2/3})$, and it is currently open to exhibit a lower bound $G = (V, E), S$ for subset preservers in which $|S| = \Omega(n^{2/3})$ and $|E| = \omega(|S|^2)$.
In other words, for $h > 1$, the current best known subset preserver lower bound that could be used is just a complete bipartite graph, and so we can do no better than the construction using bipartite cores.
\end{remark}

\subsection{Improvement with Spanner Lower Bounds}

We next prove Theorem \ref{thm:gap-pair-lb}.
%\fabr{This should be Thr \ref{thm:gap-pair-lb}}
%First we show:
%\begin{theorem} \label{thm:linear-lb-err}
%There are absolute constants $\eps, \delta > 0$ such that any $+n^{\delta}$-approximate $2$-FT preserver for a single $(s, t)$ pair needs $\Omega(n^{1 + \eps})$ edges in the worst case.
%\end{theorem}
Our new ``inner graph'' that replaces the bipartite core is drawn from the following lemma\footnote{The result proved in \cite{AmirGreg} is more general than this one; the parameters have been instantiated in this statement to suit our purposes.}:
\begin{lemma} [\cite{AmirGreg}] \label{lem:43-spanner-lb}
There are absolute constants $\eps, \delta > 0$, a family of $n$-node graphs $G = (V, E)$, node subsets $S, T \subseteq V$ of size $|S| = |T| = \Theta(n^{1/2 - \delta})$, and a set $P \subseteq S \times T$ such that any subgraph $H \subseteq G$ on $o(n^{1 + \eps})$ edges has $\dist_H(s, t) > \dist_G(s, t) + n^{\delta}$ for some $(s, t) \in P$.
Moreover, we have $\dist_G(s, t) = L$ for all $(s, t) \in P$, and $\dist_G(s, t) \ge L$ for all $(s, t) \in S \times T$.
\end{lemma}
We now describe our construction.
First, as before, we take trees $\tree_s, \tree_t$ which are copies of $\tree^1$ rooted at $s, t$ respectively.
We now label leaves of $\tree_s$ (and $\tree_t$) as \emph{good leaves} or \emph{bad leaves} using the following iterative process.
Arbitrarily select a leaf $\ell$ and label it a \emph{good leaf}.
Next, for all leaves $\ell'$ satisfying $\dist_{\tree_s}(s, \ell') \in [\dist_{\tree_s}(s, \ell) - n^{\delta}, \dist_{\tree_s}(s, \ell) + n^{\delta}]$, we label $\ell'$ a bad leaf.
We then arbitrarily select another good leaf from among the unlabelled leaves, and repeat until all leaves have a label.
Note that we have $\Theta(n^{1/2})$ leaves of $\tree_s$; by construction $\dist_{\tree_s}(s, \ell) \ne \dist_{\tree_s}(s, \ell')$ for any two leaves $\ell, \ell'$, and so the total number of good leaves is $\Theta(n^{1/2 - \delta})$.
Note:
\begin{lemma} [Compare to Lemma \ref{lem:faultyT}] \label{lem:faultyT-good}
One has that any \textbf{good} leaf $\ell^h_j$ is the leaf at minimum finite distance $d':=\dist_{\tree^h - F^h_j}(rt(\tree^h),\ell^h_j)$ from $rt(\tree^h)$ in $\tree^h- F^h_j$, and any other \textbf{good} leaf in $L(\tree^h)-\{\ell^h_j\}$ is at distance at least $d'+n^{\delta}$ from $rt(\tree^h)$. 
\end{lemma}
\begin{proof}
Immediate from Lemma \ref{lem:faultyT} and the selection of good leaves.
\end{proof}

We insert a graph $H$ drawn from Lemma \ref{lem:43-spanner-lb} into the graph $\tree_s \cup \tree_t$, using the good leaves of $\tree_s$ as the set $S$ and the good leaves of $\tree_t$ as the set $T$ (as before, all other nodes in $H$ are newly added to the graph in this step).
Note that the final graph $\tree_s \cup H \cup \tree_t$ still has $N = \Theta(n)$ nodes.
This completes the construction.

We now argue correctness, i.e. we show that one cannot sparsify the final graph to $o(N^{1 + \eps})$ edges without introducing $+n^{\delta}$ error in the $s$-$t$ distance for some well-chosen set of two faults.
The proof is essentially identical to the one used above, but we repeat it for completeness.

\begin{proof} [Proof of Theorem \ref{thm:gap-pair-lb}]
Let $G' = (V, E')$ be a subgraph of the final graph $G = (V, E)$ with $|E'| = o(N^{1 + \eps})$ edges.
By Lemma \ref{lem:43-spanner-lb}, there is a pair of good leaves $(\ell_s, \ell_t) \in P$ such that
$$\dist_{G'[H]}(\ell_s, \ell_t) > L + n^{\delta}$$
(where $G'[H]$ denotes the copy of $H$ in $G'$).
By Lemma \ref{lem:faultyT-good}, there are faults $f_1, f_2$ such that $\ell_s$ is the leaf at minimum distance $d'_s$ from $s$, $\ell_t$ is the leaf at minimum distance $d'_t$ from $t$, $(\ell_s, \ell_t) \in P$, and the distance from $s$ (resp. $t$) to any other good leaf $\ell'_s \ne \ell_s$ (resp. $\ell'_t \ne \ell_t$) is at least $d'_s + n^{\delta}$ (resp. $d'_t + n^{\delta}$).
Thus, we have
$$\dist_{G \setminus \{f_1, f_2\}}(s, t) \le d'_s + L + d'_t.$$
We now lower bound this distance in $G'$.
As before, there are two cases: either the shortest $(s, t)$ path in $G'$ traverses a shortest $(\ell_s, \ell_t)$ path in $G'[H]$, or it does not.
If so, then we have
$$\dist_{G' \setminus \{f_1, f_2\}}(s, t) \ge d'_s + (L + n^{\delta}) + d'_t \ge \dist_{G \setminus \{f_1, f_2\}}(s, t) + n^{\delta}$$
and so $G'$ is not a $2$-FT $+n^{\delta}-1$ $(s, t)$ preserver of $G$, and the theorem follows.
Otherwise, if the shortest $(s, t)$ path in $G'$ does \emph{not} traverse a shortest $(\ell_s, \ell_t)$ path in $G'[H]$, then by construction it passes through (w.l.o.g.) some good leaf $\ell'_s \ne \ell_s \in \tree_s$ and $\ell'_t \in \tree_s$ (where $\ell'_t$ is possibly equal to $\ell_t$).
We then have
\begin{align*}
\dist_{G' \setminus \{f_1, f_2\}}(s, t) &\ge \dist_{G' \setminus \{f_1, f_2\}}(s, \ell'_s) + \dist_{G \setminus \{f_1, f_2\}}(\ell'_s, \ell'_t) + \dist_{G' \setminus \{f_1, f_2\}}(\ell'_t, t)\\
&\ge (d'_s + n^{\delta}) + L + d'_t\\
&\ge \dist_{G \setminus \{f_1, f_2\}}(s, t) + n^{\delta}
\end{align*}
and the theorem follows.
\end{proof}

Finally, by tolerating one additional fault, we can obtain a strong incompressibility result, hence proving Theorem \ref{thm:gap-compression-lb}:
%\vspace{-4pt}
%\begin{theorem} %\label{thm:gap-compression-lb}
%There are absolute constants $\eps, \delta > 0$ such that any $+n^{\delta}$-approximate $3$-FT distance sensitivity oracle for a single $(s, t)$ pair uses $\Omega(n^{1 + \eps})$ bits of space in the worst case.
%\end{theorem}
The proof is \emph{nearly} identical to the above, but there are two key differences.
First, we use unbalanced trees: $\tree_s$ is a copy of $\tree^1$ while $\tree_t$ is a copy of $\tree^2$.
Hence, there are three total faults in the sets $F^s_{\ell_s} \cup F^t_{\ell_t}$ used to ``select'' the appropriate leaves $s, t$.
We define good leaves exactly as before.
We use a slightly different lemma for our inner graph (which is proved using the same construction):
\begin{lemma} [\cite{AmirGreg}] \label{lem:43-spanner-lb-inc}
There is an absolute constant $\delta > 0$, a family of $n$-node graphs $G = (V, E)$, node subsets $S, T \subseteq V$ of size $|S| = \Theta(n^{1/2 - \delta}), |T| = \Theta(n^{2/3 - \delta})$, and a set $P \subseteq S \times T$ of size $|P| = |S||T|n^{-o(1)}$ with the following property: for each pair $(s, t) \in P$, we may assign a set of edges in $G$ to $p$ such that (1) no edge is assigned to two or more pairs, and (2) if all edges assigned to a pair $(s, t)$ are removed from $G$, then $\dist_(s, t)$ increases by $+n^{\delta}$.
Moreover, $\dist(s, t) = L$ for all $(s, t) \in P$, and $\dist(s, t) \ge L$ for all $(s, t) \in S \times T$.
\end{lemma}

%To show incompressibility, we now argue:
\begin{proof} [Proof of Theorem \ref{thm:gap-compression-lb}]
Let $G, P$ be a graph and pair set drawn from Lemma \ref{lem:43-spanner-lb-inc}.
Define a family of $2^{|P|}$ subgraphs by independently keeping or removing all edges assigned to each pair in $P$.
We will argue that any $n^{\delta/2}$-\fab{additive} distance sensitivity oracle must use a different representation for each such subgraph, and thus, $|P| = \Omega(n^{1 + \eps})$ bits of space are required in the worst case.

Suppose towards a contradiction that a distance sensitivity oracle uses the same space representation for two such subgraphs $G_1, G_2$, and let $(\ell_s, \ell_t) \in P$ be a pair for which its owned edges are kept in $G_1$ but removed in $G_2$.
By an identical argument to the one used in Theorem \ref{thm:gap-pair-lb}, we have
$$\dist_{G_1 \setminus \{f_1, f_2, f_3\}}(s, t) \le d'_s + L + d'_t$$
for fault set $\{f_1, f_2, f_3\} = F^s_{\ell_s} \cup F^t_{\ell_t}$ (where $d'_s, d'_t$ are defined exactly as before).
Meanwhile, also by the same argument used in Theorem \ref{thm:gap-pair-lb}, we have
$$\dist_{G_2 \setminus \{f_1, f_2, f_3\}}(s, t) \ge d'_s + L + d'_t + n^{\delta} \ge \dist_{G_1 \setminus \{f_1, f_2, f_3\}}(s, t) + n^{\delta}.$$
Since $G_1, G_2$ are stored identically by the distance sensitivity oracle, it must answer the query $\{f_1, f_2, f_3\}$ identically for both graphs.
However, since the right answer differs by $+n^{\delta}$ from $G_1$ to $G_2$, it follows that the oracle will have at least $+n^{\delta}/2-1$ error on one of the two instances.
\end{proof}

\subsection{Lower Bound for $S\times T$ Preservers}
\label{sec:LB_ST}

\begin{theorem}
\label{thr:lowerboundST}
For every positive integer $f$, there exists a graph $G=(V,E)$ and subsets $S,T \subseteq V$, such that every $(2f)$-FT $1$-additive $S \times T$ spanner (hence $S \times T$ preserver) of $G$ has size $\Omega(|S|^{1/(f+1)} \cdot  |T|^{1/(f+1)} \cdot (n/f)^{2-2/(f+1)})$.
\end{theorem}
\begin{proof}
The graph $G$ is constructed as follows (see also Figure \ref{fig:ST}). For each $s_i\in S$, we construct a copy $\tree_{s_i}$ of $\tree^f$ rooted at $s_i$ with size parameter
$$d_S=\left(\frac{n}{3(f+1)|S|}\right)^{\frac{1}{f+1}}-1.$$
Similarly, for each $t_j\in T$, we construct a copy $\tree_{t_j}$ of $\tree^f$ rooted at $t_j$  with size parameter
$$d_T=\left(\frac{n}{3(f+1)|T|}\right)^{\frac{1}{f+1}}-1.$$
Finally, we add a complete bipartite graph between the leaves of each $\tree^f_{s_i}$ and the leaves of each $\tree^f_{t_j}$. We call the edges of the last type the bipartite core of $G$.

Note that by Lemma \ref{lem:heightT-nodesT} the total number of nodes is $n$. Furthermore, the bipartite core has size
$$
|S|\, |T|\, d_S^{f} \, d_\tree^{f} = \Omega\left(|S|\, |T|\, \left(\frac{n^{2}}{9(f+1)^2|S||T|}\right)^{\frac{f}{f+1}}\right)=\Omega\left(|S|^{\frac{1}{f+1}}|T|^{\frac{1}{f+1}}\left(\frac{n}{f}\right)^{2-\frac{2}{f+1}}\right)  
$$
The rest of the proof follows along the same line as in Lemma \ref{lem:Sh}: given any edge $e=(\ell_{s_i},\ell_{t_j})$ between a leaf $\ell_{s_i}$ of $T_{s_i}$ and a leaf $\ell_{t_j}$ of $T_{t_j}$, removing $e$ would cause an increase of the stretch between $s_i$ and $t_j$ by at least an additive $2$ for a proper choice of $f$ faults in both $\tree_{s_i}$ and $\tree_{t_j}$.
\end{proof}

\end{document}